\documentclass[12pt, draftclsnofoot, onecolumn]{IEEEtran}

\usepackage[ruled,vlined,algo2e]{algorithm2e}
\usepackage{amsmath, amssymb}
\usepackage{color}
\usepackage{cite}
\usepackage{caption}
\usepackage{tikz}
\usetikzlibrary{matrix,shapes,arrows,positioning}
\usepackage{booktabs}

\usepackage{url}
\usepackage{amsmath}
\usepackage{amsthm}
\usepackage{algorithmicx, algorithm ,algpseudocode}
\usepackage{verbatim}
\usepackage{graphicx}
\usepackage{subfigure}

\newtheorem{theorem}{Theorem}

\newtheorem{definition}{Definition}

\newtheorem{proposition}{Proposition}

\newtheorem{prob}{Problem}

\long\def \ignoreme#1{}

\begin{document}

\newtheorem{myprop}{Property}
\title{Leveraging Social Communities for Optimizing Cellular Device-to-Device Communications}

\author{Md Abdul Alim, Tianyi Pan, My T. Thai, and Walid Saad 
\IEEEcompsocitemizethanks{\IEEEcompsocthanksitem M. A. Alim, T. Pan, M. T. Thai are with the Department of Computer and Information Science and Engineering, University of Florida, Gainesville, FL 32611 (email: \{alim,tianyi,mythai\}@cise.ufl.edu).
\IEEEcompsocthanksitem W. Saad is with Wireless@VT, Bradley Department of Electrical and Computer Engineering, Virginia Tech, Blacksburg, VA. (email: walids@vt.edu).}
}

\maketitle


\begin{abstract}\label{label:abstract}
Device-to-device (D2D) communications over licensed wireless spectrum has been recently proposed as a promising technology to meet the capacity crunch of next generation cellular networks. However, due to the high mobility of cellular devices, establishing and ensuring the success of D2D transmission becomes a major challenge. To this end, in this paper, a novel framework is proposed to enable devices to form  multi-hop D2D connections in an effort to maintain sustainable communication in the presence of device mobility. To solve the problem posed by device mobility, in contrast to existing works, which mostly focus on physical domain information, a durable community based approach is introduced taking social encounters into context. It is shown that the proposed scheme can derive an optimal solution for time sensitive content transmission while also minimizing the cost that the base station pays in order to incentivize users to participate in D2D. Simulation results show that the proposed social community aware approach yields significant performance gain, in terms of the amount of traffic offloaded from the cellular network to the D2D tier, compared to the classical social-unaware methods.
\end{abstract}


\section{Introduction}\label{label:lintro}
The demand for wireless data services has increased exponentially in the past decade thus straining the capacity of existing wireless cellular networks \cite{han12} and \cite{fodor12}. One promising solution to meet this capacity crunch is to offload cellular traffic via the use of direct device-to-device (D2D) communications for enabling proximity services over the cellular licensed band \cite{pei2013resource}. To reap the benefits of D2D over cellular, there is a need to optimize and manage the added cellular interference resulting from D2D \cite{madan2010cell}. However, due to the high mobility of cellular devices, establishing and ensuring the success of D2D transmission is a major challenge. 

Recently, there has been an increased interest to operate D2D over cellular using multi-hop transmissions (henceforth referred to as \emph{multi-hop D2D}) \cite{lin2000multihop,kaufman2013spectrum,
asadi2014survey}. Such multi-hop D2D architectures can reduce the outage probability while potentially increasing the capacity of D2D communication by alleviating the effect of interference from the cellular users \cite{ma2012distributed,lee2012performance,vanganuru2012systemmilcom}. Unlike multi-hop ad hoc networks, which do not use the cellular spectrum and do not require any infrastructure, multi-hop D2D is controlled centrally by the base station (BS) for ensuring the QoS of both the cellular and D2D users simultaneously. In cellular multi-hop D2D scenarios, one must properly group the mobile devices in order to achieve the required quality-of-service (QoS). Such a grouping is particularly dependent on the mobility patterns of the devices. One major challenge in the analysis of such mobile, multi-hop D2D pertains to its strong dependence on dynamic human behavior which must be correlated with the complex QoS considerations of the cellular system.

For establishing D2D connections, the cellular BS must provide proper incentives to the users so that they become willing to share their resources for each others transmissions which in turn incurs cost to the BS \cite{saadincentiveJSAC}. Naturally, if most users are unwilling to participate in D2D transmission, the resources cannot be fully utilized, and the operation of the underlaid cellular D2D links will be jeopardized. For real-time content transmission, that must meet stringent latency requirements, a high mobility of the devices will disrupt an ongoing D2D session. This will eventually lead the D2D transmission to fail in delivering the content within the needed time bound. In such cases, the BS must initiate resource consuming cellular connection after dropping the interrupted session, thus reducing the overall network QoS and failing to exploit the benefits of D2D. Consequently, to enable reliable delivery of real-time content over multi-hop D2D at minimum BS cost, it is imperative to identify a set of reliable devices. Also, such devices must remain within the transmission range of one another during the D2D session to maintain the QoS. 

Despite significant research on cellular D2D \cite{fodor12,pei2013resource,madan2010cell}, there are very few works which consider the cellular multi-hop D2D case. One of the earliest related works is \cite{ma2012distributed} in which the relay selection problem for cellular D2D was studied. In \cite{wang2012wicom}, the authors consider D2D communication for relaying user equipment (UE) traffic while introducing a relay selection rule based on interference constraints. The works in  \cite{lee2012performance} and \cite{vanganuru2012systemmilcom} investigate the maximum ergodic capacity and outage probability of cooperative relaying in relay-assisted D2D communication. The results show that multi-hop D2D lowers the outage probability and improves cell edge throughput capacity by reducing the effect of interference from the cellular users. However, none of these works factors in the impact of mobility of devices on the system performance and on the successful delivery of time sensitive contents in particular. 

Recently, it has been observed that cellular devices carried by humans exhibit a peculiar pattern with respect to their \emph{physical encounters} both in space and time \cite{hui2011bubble} and \cite{cho2011friendship}. Such social encounters have been shown to exhibit a community structure property which implies that the network can be divided into groups of nodes with dense connections inside each group and fewer connections across groups. From a D2D perspective, users who encounter one another frequently will be likely to form a social community \cite{hui2008human}. Additionally, the longer a device stays close to another device, the mutual interaction between them grows  further compared to other sporadic contacts. Moreover, a large number of longer duration contacts over a period of time makes the mutual connection more reliable for the continuity of a D2D session which forms the basis of \emph{durable communities}. Leveraging such durable communities for improving D2D transmission constitutes therefore an opportunity that has hitherto not been explored.

The main contribution of this paper is to introduce a new framework that exploits durable social communities to enable successful transfer of a content between two devices with minimum cost using multi-hop D2D. We model the problem as a cost-effective device selection strategy on multi-hop D2D for real-time content delivery. We first formulate the durable community structure and introduce the concept of \emph{sustainable} and \emph{bridge} edges by exploiting the historical encounters of devices. We further propose a novel community detection method based on those previous encounters. Subsequently, we formulate the device selection problem as an optimization problem and we introduce an efficient method for finding the optimal set of devices on multi-hop path leveraging those social communities. This is in contrast to most existing works on multi-hop D2D that solely focus on system performance \cite{lin2000multihop,kaufman2013spectrum,
asadi2014survey,ma2012distributed,lee2012performance,vanganuru2012systemmilcom}. Simulation results show that our method outperforms classical social-unaware methods significantly on traces generated by the state-of-the-art mobility models. Note that, unlike the more classical case of delay tolerant networks (DTNs) \cite{hui2008human}, we consider only time sensitive content transfer between source and target with certain delay constraint on the total transmission time. This makes our  D2D transmission fundamentally different than the DTN which is opportunistic in nature. In addition, signal interference, resource allocation, noise and fading are intrinsic design parameters in D2D communication underlaying cellular networks which makes the design and operation of D2D completely different from DTNs and related ideas such as ad hoc networks.

The rest of the paper is organized as follows. Section \ref{label:model} introduces the system while Section \ref{label:problemFormulation} provides the problem formulation. Section \ref{label:lcr} discusses reliable device selection procedure for multi-hop D2D. Simulation results are analyzed in Section \ref{label:experiment} and conclusions are drawn in Section \ref{label:conclusion}.

\section{System Overview and Model Representation}\label{label:model}
\subsection{System Overview}
Consider the downlink transmission of an OFDMA cellular network consisting of a single base station (BS) and a set $\mathcal{N}$ of user equipments (UEs). The UEs are able to communicate with one another using D2D links that are underlaid on the cellular network. The total bandwidth $\mathcal{B}$ is divided into $F$ resource blocks (RB) in the set $\mathcal{F}$. We consider a co-channel network deployment in which $\mathcal{B}$ is shared between cellular and D2D transmissions while considering one RB per UE. We assume UE $i$ requests a content from BS which, in turn, selects UE $j$ ($i,j \in \mathcal{N}$), among other UEs having the content, as the source. The BS will enable direct D2D connections between UE $i$ and UE $j$ when the distance between them is within a desired D2D communication range $d_{max}$ which, in turn, corresponds to a required signal-to-interference-plus-noise ratio (SINR) as shown in Fig. \ref{fig:direct}. 

In practice, setting up reliable direct D2D connections while satisfying the quality-of-service (QoS)
requirements of both the traditional cellular UEs (CUEs) as well as the D2D UEs is challenging. On the one hand, the unreliable propagation medium and longer distance might affect the link quality between D2D devices (Fig. \ref{fig:dmax}). On the other hand, interference from other cellular and D2D UEs sharing the same RB will also contribute toward lowering the SINR (Fig. \ref{fig:cin}). In such low SINR cases, the use of multi-hop D2D communications can be beneficial to enhance the overall D2D QoS.

\begin{figure}
  \centering  
  \subfigure[D2D transmission with high SINR due to distance $d \le d_{max}$] {
		\includegraphics[scale=0.35]{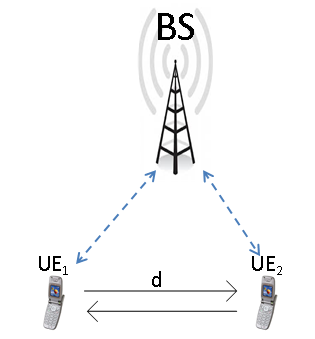}
		\label{fig:direct}
	}\hspace{1.5em}
	\subfigure[D2D transmission not possible due to low SINR as $d > d_{max}$ ] {
		\includegraphics[scale=0.40]{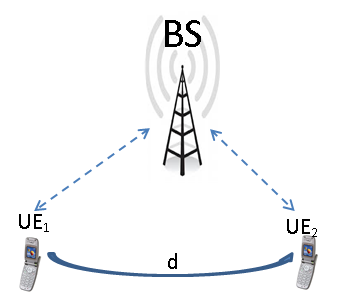}
		\label{fig:dmax}
	}\hspace{1.5em}
  \subfigure[Channel interference between cellular communication (UE$_3$ and BS) and D2D pairs (UE$_1$, UE$_2$) and (UE$_4$, UE$_5$)] {
		\includegraphics[scale=0.45]{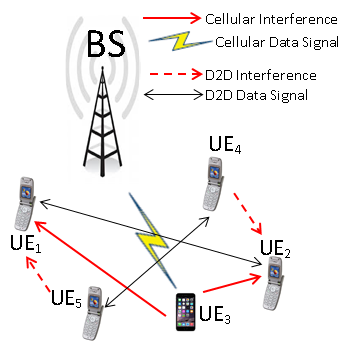}
		\label{fig:cin}
	}
  \caption{D2D communication scenario before the transmission takes place}
	\label{fig:b4scenario}
\end{figure}

Indeed, the effectiveness of multi-hop D2D depends on suitable device selection mechanisms. Ideally, for the D2D to successfully sustain data transmission, the devices that are chosen along the multi-hop D2D path must not move beyond the D2D range during a communication session so as to maintain the desired SINR target. Designing such mechanisms is challenging due to the coupling between mobility patterns, incentives for sharing resources, and network QoS. In our model, we focus on selecting a \emph{least cost reliable multi-hop path} for real-time content delivery from a source to a destination. It has been observed that mobility and physical encounter patterns are very closely related to social structures, and very often frequency and length of physical interaction is strongly correlated with proximity \cite{cho2011friendship}. 
Therefore, we leverage the historical encounter patterns of devices to identify social communities that gives indication on how devices come closer to each other. Thus, the goal of the proposed least cost multi-hop path approach is to select devices based on the social encounters and communities so as to make sure they stay within close proximity of one another during the D2D session. 

\begin{figure*}[!ht]
  \centering  
  	\includegraphics[scale=0.48]{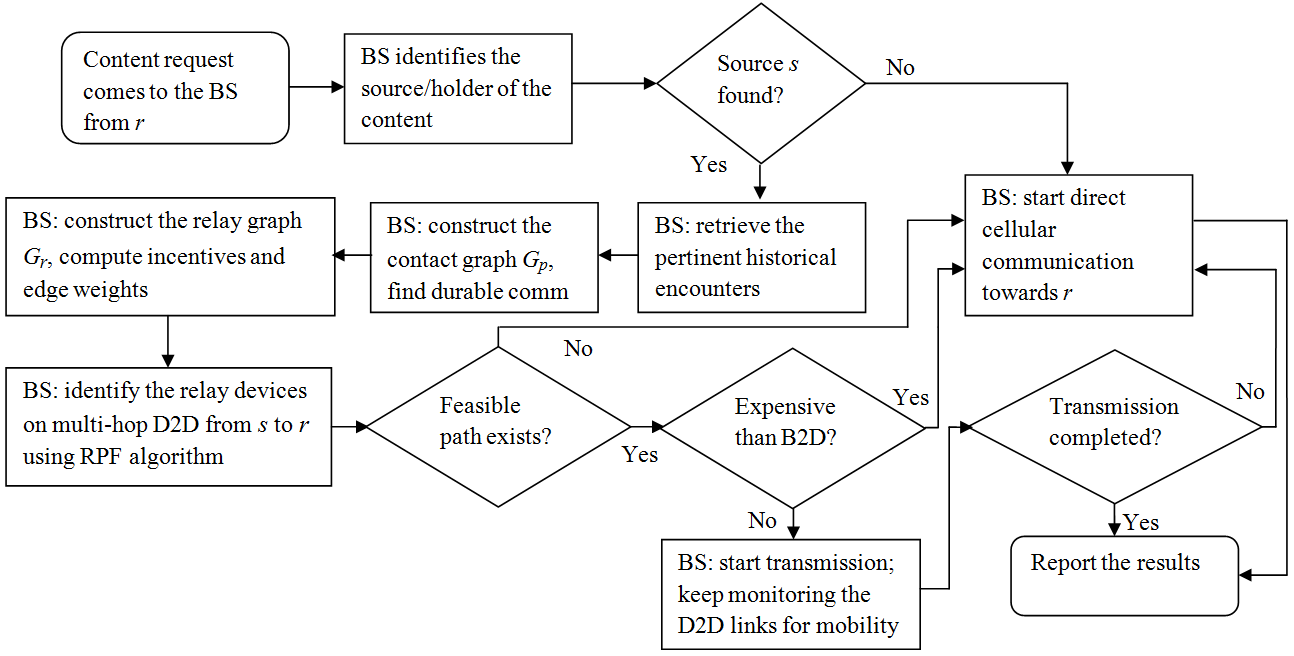}
	\caption{Flow chart for the proposed solution scheme}
	\label{fig:solution_flow}
\end{figure*}
A flow chart that summarizes the implementation of the proposed scheme is shown in Fig. \ref{fig:solution_flow}. Whenever a request for a content comes to the BS from a device $r$, the BS identifies the source of the content. If no such device is found to hold the content, then, the content is transmitted directly from the BS towards $r$ using cellular communication. 

If a device $s$ having the content is identified, the BS initiates the durable community detection phase by invoking the DCD algorithm that is detailed in Subsection~\ref{ss:dci}. The BS then assigns proper edge weights to each of the D2D pairs present in its coverage area using the social-based technique that is  explained in Subsection~\ref{wijassignment}. Finally, the BS identifies the multi-hop D2D path to relay the content from $s$ to $r$, if there exists any such feasible path that can deliver the content within a certain time threshold $t_{max}$, instead of; otherwise, the BS initiates a direct cellular connection towards $r$. In the former case (when a feasible path exists), if the total incentive that the BS has to pay to the relay devices on the multi-hop path is larger than the direct BS to $r$ cost which is termed as \emph{B2D cost}, the BS also initiates a direct cellular connection towards $r$ rather than serving the content via D2D. Once the content transmission starts via multi-hop D2D, the BS keeps track of the pairwise mobility of devices for each hop. If the device mobility leads to a minimum allowable SINR that is below a certain threshold, the multi-hop D2D connection can no longer be sustained. At this point, the BS  has to initiate a direct cellular connection towards $r$ to fulfill its content request. Next, we describe the necessary system model.

\subsection{System Model}
In our network, we consider real-time content sharing among mobile D2D users with strict delay requirements. We assume device $r$ requests a content of size $b$ from the BS at time $t$. The BS identifies $s$, another UE, as the peer device having the content that would serve the request of $r$ via D2D. There are several approaches to identify a suitable source for a requested content in literature \cite{wang2014game} which is not our focus in this paper. Hereinafter, $s$ is referred to as the source device and $r$ as the destination device. However, as discussed previously, these UEs may not be able to communicate directly due to physical constraints and hence a multi-hop path needs to be identified for effective content transfer.

\begin{center}
  \begin{table*}[t]
		\caption{Summary of Important Symbols}
		\centering
    \begin{tabular}{ | l | l | l | l |}
    \hline
    Symbol & Description & Symbol & Description \\		
	\hline			
					$\mathcal{N}$ & Set of UEs in the network &					 $G_r$ & Multi-hop D2D graph at time $t$	\\
			$\mathcal{B}$ & Bandwidth of the network & 					$\eta$ & Speed of light \\												
			$\mathcal{F}$ & Set of RBs &							$c_{i,j}$ & Cost that BS pays to incentivize $i$ to send content to $j$ \\						
			$F$ & Total number of RBs &							$\psi$ & Shadowing component \\		
			$d_{max}$ & Maximum D2D range &							${\Delta}t$ & Historical encounter span\\	
			$s$ & Source of the content &							$t$ & Actual time when content request is generated from $r$\\	
			$r$ & Destination/requester of the content &					$G_p$ & Contact graph\\				
			$b$ & Content size &								$t_c$ & Time that would have taken to transmit content $c$ of \\				
			& & & size $b$ from $s$ to $r$ if they were within $d_{max}$\\
			$R_{i,j}$ & Achievable data rate between device $i$ and $j$ &			$D_{ij}$ & Encounter duration between device $i$ and $j$\\	
			$l_z$ & Bandwidth of resource block $z\in\mathcal{F}$ &				$\delta$ & Predefined stability threshold\\	
			$\mathcal{Z}$ & Set of devices sharing same RB $z$ &				$L_i(t)$ & Position of device $i$ at time $t$\\	
			$g_{i,j}$ & Channel gain between device $i$ and $j$ &				$\bar{D}_{ij}$ & Average contact duration between device $i$ and $j$ in ${\Delta}t$\\	
			$\sigma^2$ & variance of the Gaussian noise &					$\lambda_{i,j}$ & Average number of encounters between $i$ and $j$\\	
			$p_i$ & Transmit power of device $i$ &						$G^e$ & Encounter history graph\\				
			$\gamma_{i,j}$ & SINR at device $j$ for the link $i \rightarrow j$ &		$\zeta$ & Strength threshold\\					
			$d_{i,j}$ & Distance between device $i$ and $j$ &				$\rho$ & Predetermined weight factor\\						
			$\alpha$ & Path loss exponent &							$\mathcal{C}$ & Set of durable communities\\	
			$m_0$ & Fading component &							$w^b_{uv}$ & Weight of bridge edge between $u$ and $v$\\			
			$p_{i,j}$ & Received power at device $j$ for the link $i \rightarrow j$ &	$w^s_{uv}$ & Weight of sustainable edge between $u$ and $v$\\					
			$t_{i,j}$ & Time required to transmit content  &			$B_{u,v}$ & Percentage of actual encounter duration larger than $t_c$\\				
			& from device $i$ to $j$ & & between device $u$ and $v$ \\
			$t^p_{i,j}$ & Propagation delay between device $i$ to $j$ &			$h_C$ & Durability of community $C$\\	
			$t^x_{i,j}$ & Transmission delay from device $i$ to $j$ &			$k$ & Number of detected communities\\

    \hline
    \end{tabular}		
		\label{tableN}
	\end{table*}
\end{center}

The achievable rate $R_{i,j}$ for the transmission between a device $i$ and device $j$ is
\begin{equation}
R_{i,j} = l_z \log_2(1+\gamma_{i,j}),
\label{rij}
\end{equation}
where $l_z$ is the bandwidth of RB $z \in \mathcal{F}$ used by $i$ for its data transmission to $j$, $\gamma_{i,j}$ denotes the SINR for $j$ from $i$. For the link between $i$ and $j$, considering signal interference from all other devices using the same RB $z$, we have
\begin{equation}
\gamma_{i,j} = \frac{p_i g_{i,j}}{\sum_{i' \in \mathcal{Z},i'\neq i}^{} p_{i'} g_{i'j} + \sigma^2},
\label{gij}
\end{equation}
where $\mathcal{Z}$ is the set of devices sharing RB $z$, $g_{i,j}$ is the channel gain between $i$ and $j$, $p_i$ is the transmit power of device $i$, and $\sigma^2$ is the variance of the Gaussian noise. Here, we note that the BS and the devices operate in a half-duplex mode and the same set of resources (i.e. subcarriers) is shared for transmission of content. In our model, devices on several D2D links can transmit simultaneously and hence can cause interference with one another when using the same RB. However, devices on different hops do not interfere with one another over the same RB. The proposed approach can accommodate any algorithm for allocating RBs to the various D2D and cellular links. Without loss of generality, hereinafter, we adopt graph coloring techniques such as in \cite{tan2011graph} to perform this assignment. In our model, in line with existing D2D works \cite{asadi2014survey} and for tractability, we do not consider interference on the reverse, acknowledgment link. We have observed in our experimental evaluation, incorporating reverse link interference into the formulation does not significantly affect the conclusions.

Due to the fact that one cannot know which D2D links will be actively relaying at every hop until we execute our proposed relay path finder (RPF) algorithm described in Section \ref{label:lcr}, we assume that all the D2D links are active. This enables us to compute the data rate of the links which is required by the RPF algorithm for choosing relays that can deliver the content within $t_{max}$. In order to reduce the interference between the cellular links and the D2D links, we identify the D2D links which are within close proximity to the cellular link and we ensure that they do not reuse the same RBs. We only allow those links which are sufficiently far apart to share the same resources. This is essentially similar to the classical frequency reuse concept used in cellular networks, but now we apply to D2D transmissions. For each D2D link $(i,j)$, we identify the interference set for this link. An interference set for $(i,j)$ contains all the links whose transmitter or receiver are within a certain distance from the transmitter $i$ of the link $(i,j)$ and could potentially cause large interference. In the graph coloring based resource allocation scheme that we use, these links are assigned different resource blocks. Links that are significantly far away from each other are allowed to have same RB. Once the RB allocation is complete, we utilize Eq. \eqref{rij} and \eqref{gij} to compute the data rates for each link. 

For the wireless network, we consider distance-dependent path loss and multipath Rayleigh fading along with log-normal shadowing. Thus, the received power of each link between devices $i$ and $j$ can be described as $p_{i,j}= p_i . (d_{i,j})^{-\alpha} . |m_0|^2 . 10^{\psi/10}$, where $p_i$ is the transmit power of device $i$, $\alpha$ is the path loss exponent, $m_0$ is the fading component, and $\psi$ is the log-normal shadowing component. 

Given this SINR model, we now formulate the time required for the transmission between device $i$ and $j$. The time $t_{i,j}$ in the link (hop) from $i$ to $j$ is defined as
\begin{equation}
t_{i,j}  = t_{i,j}^p + t_{i,j}^x = \frac{d_{i,j}}{\eta}+\frac{b}{R_{i,j}},
\label{tij}
\end{equation}
where $t_{i,j}^p$ is the propagation delay between device $i$ and device $j$ which, in turn, depends on  the distance $d_{i,j}$ of the single hop link between $i$ and $j$, and the speed of light $\eta$. The transmission delay, $t_{i,j}^x$, depends on the packet size $b$ and on the achievable data rate for the transmission between $i$ and $j$ as per (\ref{rij}).

To incentivize a certain device $i$ for sharing its resources with another device $j$, the BS must incur a cost $c_{i,j}$. A device that experiences a good channel and that has a higher transmit power will be able to transfer content more efficiently than others, and hence is a better candidate for D2D from the BS's perspective. Accordingly, we have,
\begin{equation}
 c_{i,j}= p_{i,j} = p_i \cdot (d_{i,j})^{-\alpha} \cdot |m_0|^2 \cdot 10^{\psi/10}.
 \label{equ:weightNode}
\end{equation}

 This incentive/cost can be in the form of monetary remunerations, coupons, or free data. We summarize most of the important notations used throughout this work in Table \ref{tableN}. Next, we define the necessary framework for formulating the problem of identifying reliable devices on multi-hop D2D.

\section{Problem Formulation and Solution}\label{label:problemFormulation}

\subsection{Problem Formulation}

Given this wireless network model, the next goal is to find a set of devices that would enable feasible multi-hop D2D communications while satisfying stringent delay constraints and minimize the BS's cost, as per \eqref{equ:weightNode}. We introduce the concept of feasible path formally as follows:

\begin{definition}(Feasible Path)
\normalfont Given a cellular network $G=(\mathcal{V,E})$, where $\mathcal{V}$ is the set of all devices and $\mathcal{E}$ is the set of links that connects them, a \emph{feasible path} from source $s$ to destination $r$ in $G$, is an ordering $P$ of devices in $\mathcal{V}$, where $P=<i_1, \ldots, i_k>$ such that $i_1=s$, $i_k=r$, $(i_j,i_{j+1}) \in \mathcal{E}$ and given the interference and mobility of devices, $\sum_{i=1, j=i+1}^{k} t_{i,j} \le t_{max}$ where $t_{i,j}$ and $t_{max}$ indicate the time required to transfer a content from $i$ to $j$ and the maximum allowed content sharing time, respectively.
\end{definition}

For successful delivery of a content using  multi-hop D2D, the devices on a \emph{feasible path} must also remain within a range that corresponds to the desired SINR throughout the D2D session. To combine these properties, we now present the cost-effective device selection problem for multi-hop D2D (CEDS-MD):

\begin{prob}(CEDS-MD)
\normalfont Cost-effective device selection for multi-hop D2D (CEDS-MD) seeks to identify a \emph{feasible path} $P$ that results in minimum cost of transmission from source $s$ to destination $r$ by minimizing the device cost denoted by $C(P)=\sum_{(i,j) \in P}^{}  c_{i,j} $, where $c_{i,j}$ is the cost of BS for incentivizing device $i$ to share resource with device $j$, and $i$ is the immediate predecessor of $j$ in the feasible path and the devices on $P$ remain within the D2D transmission range throughout $t_{max}$ as governed by the cellular base station.
\label{def:lcdtA}
\end{prob}
\subsection{Social Community Aware Cellular Network}
Incorporating social based device proximity information with conventional
physical layer metrics enables better resource utilization and enhanced traffic offload in D2D \cite{proebster2012context}. 
However, these measures are  not able to capture the impact of user mobility on the successful completion of D2D transmission particularly when devices are moving rapidly during the transmission. Consequently, there is a need to adopt a more realistic view for the social context by basing it on other social dimensions such as the \emph{actual encounters between users}. Device encounters have been shown to satisfy the community structure property \cite{cho2011friendship} and thus, the stability of D2D session must be correlated with durable social communities.

Therefore, as a first step towards solving CEDS-MD, we must identify \emph{durable social communities} based on the previous encounter histories. When two devices $i$ and $j$ are within the transmission range $d_{max}$ of each other, they can communicate in D2D mode under the control of the BS. Assuming a content request is generated at a given time $t$ during a day, the BS extracts all the specific historical encounters that start around $t$ in order to realistically predict the mobility pattern of the devices. To this end, the BS constructs a physical contact graph $G_p$ which is a weighted undirected graph and detects the durable communities. Devices belonging to the same community are more likely to have longer contact duration and, hence, they will get more priority to be chosen on the multi-hop D2D if they happen to be within each others proximity at content request time $t$.

In $G_p$, each edge represents the average duration of contact between two devices for a certain span ${\Delta}t$ of previous days. ${\Delta}t$ can be any number of previous days (or hours) depending on the way the encounter histories are being preserved in the BS. If $t_c$ is the time required for the content to be transmitted from $s$ to $r$ when they are within the range $d_{max}$, the BS will need to consider those previous encounters in ${\Delta}t$ that have an average duration of at least $t_c$. Although encounters having duration at least $t_c$ are good candidates for reliable connections, the longer the duration the more durable it is. To put the duration length into perspective, we not only take into account the encounters having duration of sufficient length ($t_c$) but also all the previous encounters with duration $D_{ij} \ge (1+\delta) t_c$ where the stability threshold, $\delta \ge 0$, is a user controlled parameter that reflects the importance of the duration length of encounters beyond $t_c$. At the same time, we also emphasize on the impact of encounter rates of two devices in ${\Delta}t$ paired with the duration. Next we will formally define the notions related to encounters.

\subsubsection{Community Structure and Durable Community}\label{label:lscd}
Now, we introduce the necessary terms to describe encounters in the context of D2D and formally define the notion of a durable community structure in this subsection. 

Assume that $i$ and $j$ come into the communication range at time $t_e$, that is, $||L_i(t_e^-) - L_j(t_e^-)|| > d_{max}$ and $||L_i(t_e) - L_j(t_e)|| \le d_{max}$, where $t_e^-$ denotes the time before $t_e$, $L_i(t)$ the position of user $i$ at time $t$, $d_{max}$ the D2D transmission range as determined by the BS and $||.||$ the distance measure. With this, we can define the D2D contact duration:
\begin{definition}
\normalfont The \emph{D2D contact duration} between users $i$ and $j$ is defined as the time during which they are in contact before moving out of the range, that is, $D_{ij} = t - t_e$ with   $\underset{t-t_e}{\operatorname{min}} \{ t: ||L_i(t) - L_j(t)|| > d_{max}, t > t_e\}$, where $t$ and $t_e$ are in the continuous time scale.
\label{def:duration}
\end{definition}

Consider a series of $q$ contact durations $D_{ij}=(D^1_{ij}, \ldots, D^q_{ij})$ between nodes $i$ and $j$ in time frame ${\Delta}t$, then, we can make the following definition:
\begin{definition}
\normalfont The \emph{average contact duration}, denoted by $\bar{D}_{ij}=\frac{\sum_{k=1}^{q}D^k_{ij}}{q}$, is the expected time during which two devices stay within $d_{max}$ before they move apart again once after coming in proximity to one another.
\label{def:avgduration}
\end{definition}

Next, let $G^e=(\mathcal{V,E}^e,\boldsymbol{T})$ be an undirected graph representing the physical encounters of $|\mathcal{V}|$ mobile devices. $\mathcal{E}^e$ is the set of undirected relationships (in this case encounters). Each edge $E^e_i$ has an associated collection of two-dimensional vectors denoted by $\boldsymbol{T}_i=(T_{i_1}, T_{i_2}, ...)$. Each element of in $\boldsymbol{T}_i$ denotes contact time and corresponding duration in ${\Delta}t$ time span, i.e., $T_{i_j} = <t_{uv},D_{uv}>$ for all the $j$ encounters between device $u$ and $v$ in ${\Delta}t$.

\emph{Contact Graph:} The request for a content is generated at time $t$ and $t_c$ is the time required to transmit the content from $s$ to $r$ if they are within range $d_{max}$. We construct an undirected and weighted contact graph $G_p=(\mathcal{V}_p,\mathcal{E}_p,\mathcal{W}_p)$, where $|\mathcal{V}_p|=n$ and $|\mathcal{E}_p|=m$. In doing so, we consider only those encounters in $G^e$ that have \emph{average contact duration} $\bar{D}_{ij}$ sufficiently long enough to cater $t_c$ starting at $t$, i.e., $\bar{D}_{ij} \ge (1+\delta)t_c$ where $\delta \ge 0$ is the predefined stability threshold. $w_{uv} \in \mathcal{W}_p$ is the weight function on each edge $(u,v) \in \mathcal{E}_p$ where $u,v \in \mathcal{V}_p$. 

\emph{Weight Assignment in $G_p$:} Encounters having average contact durations larger than $t_c$ are very good candidates for sustainable D2D transmission. However, considering only the average duration might result in choosing some encounters having a large number of less than $t_c$ duration which will negatively impact the reliable device selection for multi-hop D2D. To account for this in assigning $\mathcal{W}_p$, we will prioritize those edges having encounters with actual duration larger than $t_c$ with more weight. To this end, we define $B_{uv}$, $0 \le B_{uv} \le 1$ that denotes the percentage of times the encounter duration was actually larger than $t_c$. Accordingly, we define the weight $w_{uv} = \rho B_{uv} \cdot \lambda_{uv}+ (1-\rho) \bar{D}_{uv}$ where $\rho$, $0 \le \rho \le 1$ is a predefined weight factor that signifies how much emphasis should be put on the average encounter duration with respect to the percentage of times the encounter duration was actually larger than $t_c$ as denoted by $B_{uv}$. To account for the encounter rate we have multiplied $B_{uv}$ by the weight factor $\lambda_{uv}$, so that the impact of frequent long duration contacts can also be captured in the edge weight. $\lambda_{uv}$ denotes the average number of encounters between $u$ and $v$ over the time period ${\Delta}t$.

Next, we will define a durable community structure that will group devices having similar contact duration together. Such a structure has special properties related to bridge and sustainable edges. In fact, 
an edge $(u,v)$ in $G_p$ is said to be \emph{bridge} edge if it has small percentage of successful contact durations $B_{uv}$ which is reflected by $w_{uv} < \zeta$ where $\zeta$ is the predefined strength threshold. A \emph{sustainable} edge $(u,v)$ is defined to have large percentage of successful contact durations $B_{uv}$ which is reflected by $w_{uv} \ge \zeta$. We denote the weight of sustainable and bridge edges as $w^s_{uv}$ and $w^b_{uv}$, respectively. We leverage these edge weights in deciding the relay devices which we describe in Subsection~\ref{wijassignment}.

Consequently, a \emph{durable community structure}, denoted by $\mathcal{C} = \{ C_1, C_2, \ldots, C_k\}$, is a collection of $k$ subsets of $\mathcal{V}$ satisfying $\cup_{i=1}^{k} C_i = \mathcal{V}$. We say that, a collection of nodes $C_i \in \mathcal{C}$ and its induced subgraph is a \emph{durable community} in $G_p$ if nodes inside $C_i$ are connected primarily through sustainable edges and nodes across communities $C_i$ and $C_j$, if connected, will have bridge edges. Next, we propose an approach to detect durable communities in $G_p$.

\subsubsection{Durable Community Detection}\label{ss:dci}
For a node $u \in \mathcal{V}_p$, let $\mathcal{A}_u$ be the set of neighbors adjacent to $u$. Moreover, let $w_u$ be the weight corresponding to this set. For any $C \subseteq \mathcal{V}_p$, let $C^{in}$ and $C^{out}$ be, respectively, the set of links having both endpoints in $C$ and the set  of links heading out from $C$. Additionally, let $w_C = \sum_{(u,v) \in C^{in}}^{} w_{uv}$, $w_C^{out} = \sum_{(u,v) \in C^{out}}^{} w_{uv}$ and $w_C^+ = w_C + w_C^{out}$. 

Given the contact graph $G_p$, we seek to find a community structure $\mathcal{C} = \{ C_1, C_2, \ldots, C_k\}$ that would strive to group sustainable edges inside a community and place bridge edges across communities. Intuitively, any grouping that maximizes the ratio of sustainable edges to bridge edges inside a community achieves our objective. Thus, we define the \emph{durability} of a community $C$ as $h_C =\frac{w_C}{w_C^+}$, and we formulate the following Durable Community Detection ({\bf DCD}) optimization problem:

\begin{align*}
	\text{maximize }& \mathcal{R}=\sum_{C\in \mathcal{C}} h_C=\sum_{C\in \mathcal{C}} \frac{w_C}{w_C^+},\\    
	\text{s.t.	}&C_i \cap C_j  = \emptyset  &\forall i,j \in \{1,2,\ldots,k\},\\
		   & \bigcup_{i=1}^{k} C_i = \mathcal{V}_p
	\end{align*}

In this formulation, the number of communities $k$ is determined by optimizing the objective function $\mathcal{R}$ and is not an input parameter. Next, we show the following properties of network communities identified by optimizing our suggested metric $\mathcal{R}$: (i) links within a community have high durability contribution and (ii) links connecting communities have low durability contribution. 

\begin{proposition}
Let $\mathcal{C}=\{ C_1, C_2, \ldots, C_k\}$ be a community structure detected by optimizing $\mathcal{R}$, links within each $C_i$ are of strong durability contribution while those connecting communities are of weak durability contribution. 
\label{pro:dcd}
\end{proposition}
\begin{proof}[Proof]
For any node $u \in \mathcal{V}_p$ and subset $S \subseteq \mathcal{V}_p$, let $w_{u,S}$ be the total weight
of all links that $u$ has towards $S$ and vice versa. By this definition, we obtain $w_u = w_{u,S} + w_{u,\mathcal{V}_p\backslash S}$. 

Consider a community $C \in \mathcal{C}$, $u \in C$ and $v \notin C$. Since $v$ is not a member of $C$, we have
\begin{align*}
	\frac{w_C}{w_C^+} > \frac{w_C+w_{v,C}}{w_C^++w_v} = \frac{w_C+w_{v,C}}{w_C^++w_{v,C}+w_{v,\mathcal{V}\backslash C}},
\end{align*}
because otherwise adding $v$ to $C$ will give a better value of $\mathcal{R}$. This equality results in 
\begin{align*}
	\frac{w_{v,C}}{w_v} < \frac{w_C}{w_C^+},
\end{align*}
which, in turn, implies that the links joining $v$ to $C$ are insignificant in terms of durability contribution with respect to the total weight of $C$ as a whole.

Similarly, for any node $u \in C$, we have 
\begin{align*}
	\frac{w_C}{w_C^+} > \frac{w_C-w_{u,C}}{w_C^+-w_u} = \frac{w_C-w_{u,C}}{w_C^+-w_{u,C}+w_{u,\mathcal{V}\backslash C}},
\end{align*}
because otherwise excluding $u$ from $C$ will give a better estimation of $\mathcal{R}$. This inequality simplifies to
\begin{align*}
	\frac{w_{u,C}}{w_u} > \frac{w_C}{w_C^+},
\end{align*}
which shows that the links joining $u$ to $C$ are of significant weight having larger durability contribution in comparison to the total internal weight of $C$. 
\end{proof}
\subsubsection{A greedy algorithm for DCD problem}
Solving the DCD problem is NP-hard as shown by a similar reduction to modularity as in \cite{brandes2008modularity}. Consequently, a heuristic approach that can provide a good solution in a timely manner is  more desirable. In this regard, we propose a greedy algorithm for the DCD problem consisting of three phases, shown in Alg. \ref{algdcd}.

The first phase, referred to as the \emph{development phase}, identifies raw communities  in the input network. Initially, all nodes are unassigned and do not belong to any community. Next, a random node is selected as the first member of a new community $C$, and consequently, new members who help to maximize $C$'s durability, $h_C$, are gradually added into $C$. When there is no more node that can improve this objective of the current community, another new community is formed and the whole process is then continued in the very same manner on this newly formed community.

Next, the \emph{augmentation phase} rearranges nodes into more appropriate communities. In the first phase, new members are added into a community $C$ in a random order. Therefore, $C$'s objective value $h_C$ can further be improved if some of its members, that reduce the total durability, are excluded. Such nodes then form singleton communities. This step requires the re-evaluation of all $C$'s members as a result. The removal of such nodes creates more cohesive communities having higher internal connectedness.

In the last phase, the \emph{refinement phase}, global stability of the whole network is re-estimated. This phase looks at the merging of two adjacent communities in order to improve the overall objective function. If two communities have a large number of mutual connections between them, it is thus more durable to combine them into one community.

\begin{algorithm}[t]
\footnotesize
  \caption{DCD algorithm}
  \KwData{Network $G_p = (V_p,E_p,W_p)$}
  \KwResult{Durable community structure $\mathcal{C}$}
  {\bf \emph{Phase I: Development Phase.}}\;\\
  Initialize $\mathcal{C} \leftarrow \emptyset$  \;\\
  Initialize $Q \leftarrow \mathcal{V}_p$  \;\\  
  \While{$\exists \text{ unassigned node } x \in Q$}
  {
    $ C \leftarrow \{x\}$\;\\
    $ Q \leftarrow Q \backslash \{x\}$    \;\\
    \While{$ \exists y \in Q \text{ such that } h_{C \cup \{y\}} > h_C$}
	{
		$y \leftarrow \underset{y \in Q}{\operatorname{argmax}} \{h_{C \cup \{y\}}\}$ 		\;\\
	    $ C \leftarrow C \cup \{y\}$\;\\
	    $ Q \leftarrow Q \backslash \{y\}$    
	}    
	$\mathcal{C} \leftarrow \mathcal{C} \cup \{C\}$
  }
  {\bf \emph{Phase II: Augmentation Phase.}}\;\\
  \For{$C \in \mathcal{C}$}
  {
	\While{$ \exists x \in C \text{ such that }  h_{C \backslash \{x\}} > h_C$}
	{ 
		$ C \leftarrow C \backslash \{x\}$\;\\
	    $\mathcal{C} \leftarrow \mathcal{C} \cup \{x\}$    
	}
  }
  {\bf \emph{Phase II: Refinement Phase.}}\;\\
  \While{$ \exists C_1,C_2 \text{ such that } h_{C_1 \cup C_2} > h_{C_1} + h_{C_2}$}
  { 
  	$(C_1,C_2) \leftarrow \underset{C_1,C_2 \in \mathcal{C}} {\operatorname{argmax}} \{h_{C_1 \cup C_2} - h_{C_1} - h_{C_2}\}$\;\\
  	$\mathcal{C} \leftarrow (\mathcal{C} \backslash \{C_1,C_2\}) \cup \{C_1 \cup C_2\}$
  }  
  Return $\mathcal{C}$
  \label{algdcd}
\end{algorithm}

The run time complexity of the \emph{development} and \emph{augmentation} phases are $\mathcal{O}(nm)$. Moreover, even though the refinement phase might take $\mathcal{O}(n^3m)$ time in the worst case scenario, we have found that the DCD algorithm computes the durable communities within milliseconds even for networks containing hundreds of nodes as reported in Table \ref{table:runtime}. Since the optimal solution takes exponential time for larger instances of the network, we use smaller values of $n$ in order to obtain results for optimal solution for comparing with the running time of DCD. We formulated the DCD problem as an integer program with quadratic constraints and objective function and solved it using CPLEX \cite{cplex} to obtain the result for optimal solution. We have reported the results of run time comparison in the Table \ref{table:time}. Clearly, the run time complexity of the optimal algorithm increases exponentially as the number of devices increases in the network, whereas DCD takes only a small amount of time on all of those cases which makes DCD suitable for real-time relay selection.

\begin{center}
\begin{table}[!htb]
\centering
	\caption{Running times in seconds for DCD} 
	\begin{tabular}{*7c}
    \toprule
	    Method &  \multicolumn{6}{c}{User count (n)}   \\
        {} & 20 &50&80& 110 & 140 & 170 \\
        \midrule
         DCD & 0.006 & 0.022 & 0.05 & 0.018 &0.27 & 0.84\\ 		
	\bottomrule
	\end{tabular}
    \vspace{-8pt}
	\label{table:runtime}
\end{table}
\end{center}

\begin{center}
\begin{table}[!htb]
\centering
    \caption{Comparison of running times in seconds} 
	\begin{tabular}{*6c}
    \toprule
	    Method &  \multicolumn{5}{c}{User count (n)}   \\
        {} & 10 &15&20& 25 & 30 \\
        \midrule
         DCD & 0.006 & 0.005 & 0.006 &0.005 & 0.009\\ 		
         Optimal & 1.68 & 6.31 & 422.73 & 1465 & 2970\\
	\bottomrule
	\end{tabular}
    \vspace{-8pt}
	\label{table:time}	    
\end{table}
\end{center}

\section{Cost-effective Device Selection}\label{label:lcr}
Once content request is generated at time $t$, the BS initiates a centralized process that encompasses two tasks. First, it constructs $G_p$ and finds out durable communities as described in previous section. In the second step, the BS selects a set of devices to solve the {\bf \emph{CEDS-MD}} problem defined in Problem \ref{def:lcdtA}. To ensure high likelihood of the successful delivery of content through D2D, the BS incorporates the social encounter based community information as described subsequently. 

\subsection{Relay Graph Construction}\label{label:rgc}

The BS initiates the second step for device selection by constructing a \emph{multi-hop D2D graph} $G_r=(\mathcal{V}_r,\mathcal{E}_r, \mathcal{W}_c, \mathcal{W}_e)$ where $\mathcal{V}_r$ is the set of devices present at time $t$. $\mathcal{W}_c$ denotes the BS cost, for any $(i,j) \in \mathcal{E}_r$, $c_{i,j}$ indicates how much incentive BS has to spend in order to make device $i$ agree to share its resources with device $j$ for relay purpose as defined in \eqref{equ:weightNode}. We put an edge between two devices $i$ and $j$ if and only if the distance between them is within the D2D communication range, that is, the SINR from $i$ to $j$ is above a certain threshold as determined by the BS. Here, the BS is also considered as part of the graph where it is represented by a vertex. The edge connecting the BS and each device has a cost that pertains to the physical channel condition between them. Since a transmitting device in a D2D pair with better channel condition is preferred from the BS's point of view, the BS will pay a higher incentive and thus, it incurs more cost to the BS which is captured in equation \eqref{equ:weightNode}. In contrast, for a direct BS to device connection,  a receiving device having better channel condition with the BS will require less physical resource blocks for the transmission which will result in a smaller B2D cost.  The BS will have to use a relatively large number of resource blocks to transmit the content within $t_{max}$ to a device which is far away from it which is essentially a device experiencing poor channel condition at the BS. Consequently, the cost for BS to that device, termed as \emph{B2D} cost, will be naturally higher than a device with better channel condition. In summary, the B2D cost can be defined to be inversely proportional to the radio channel condition from the BS to that device as denoted below.
\begin{equation}
 c_{BS,j}= \frac{K}{p_{BS,j}} = K \times \{p_{BS} . (d_{BS,j})^{-\alpha} . |m_0|^2 . 10^{\psi/10}\}^{-1}.
 \label{equ:BSweight}
\end{equation}

A device located closer to the BS essentially experiences better channel condition at the BS and incurs less B2D cost to receive the content. The inverse of the numerical value of the received signal at device $j$ from the BS, denoted by $p_{BS,j}$ is a large number; the constant $K<1$ is thus required to normalize the cost so that the B2D cost is in the same scale with the multi-hop D2D cost. To account for the mobility  of the devices on the multi-hop path, i.e., increasing the likelihood of successful content delivery, we resort on identified \emph{durable communities} for the assignment of edge weight $\mathcal{W}_e$ described below. 

\subsection{$\mathcal{W}_e$ weight assignment in $G_r$}\label{wijassignment} Since the durable communities are constructed based on physical encounter history, users belonging to the same community have strong connections internally that not only help in reliable content transfer but also lay the basic foundation for stable and sustainable encounter predictions. The  BS follows specific rules in order to assign proper edge weights $W_{ij}$ between two devices $i$ and $j$ who are within $d_{max}$ in $G_r$ according to their membership in the durable communities obtained from the contact graph $G_p$. (i) Devices belonging to same community as well as connected via sustainable edge will have small weight that is inversely proportional to the total internal edge weight of that community. (ii) Devices belonging to same community but either connected with a bridge edge in $G_p$ or without any edge in $G_p$ will have larger weight in $G_r$ compared to case (i). (iii) If devices belong to different communities $C_i$ and $C_j$ and there is no edge connecting them in $G_p$ or the edge connecting them is a bridge edge in $G_p$, the edge connecting them in $G_r$ will have large weight that is inversely proportional to the weight of the edge bearing minimum weight among all edges connecting $C_i$ and $C_j$ in $G_p$. If there is no edge connecting $C_i$ and $C_j$ in $G_p$, we assign $W_{ij}$ the value which is the maximum weight between any two devices in $G_r$. (iv) If devices belong to different communities $C_i$ and $C_j$ and a sustainable edge connects them in $G_p$, the edge weight $W_{ij}$ between $i$ and $j$ in $G_r$ will be smaller than that of case (iii). According to these four criteria, edge weights are assigned between adjacent devices (within $d_{max}$) in $G_r$ which help our proposed solution RPF to choose suitable relay devices for multi-hop content transfer as we will demonstrate in the performance evaluation section.

\subsection{Social Community Aware Device Selection for Multi-hop D2D}\label{label:lcde}

The goal is to find a least cost path from $s$ to $r$ in relay graph $G_r$ within practical constraints of maximum delivery time imposed as part of latency which puts a limit on the number of relay devices. At the same time, we emphasis on the importance of incorporating durable communities into decision making process of device selection for successful D2D session completion. To take this into account, we modify the cost of the path $P$ in Problem \ref{def:lcdtA} as part of our solution to \emph{CEDS-MD}. Accordingly, we include the edge weight $W_{ij}$ that was computed in Section \ref{wijassignment}, to obtain the total cost $w_{ij}$ between $i$ and $j$ as follows:
\begin{equation}
w_{ij} = W_{ij} + c_{i,j}.
\label{equ:weightR}
\end{equation}

Note that, both the terms in the right hand side of \eqref{equ:weightR} are normalized and of the same order of magnitude. For real-time content sharing with D2D communication, we can formulate the optimal relay selection problem in multi-hop D2D cellular network as the following optimization problem. Let the variable $x_{ij}$ represent each edge $(i,j) \in E_r$:
\begin{equation}
    x_{ij}=\begin{cases}
    1, & \text{if $e(i,j)$ is selected for least cost feasible path}.\\
    0, & \text{otherwise}.
  \end{cases}
\label{z_i}  
\end{equation}

\vspace{1.0cm}
We have the following Integer Program (IP):

\begin{eqnarray}
	\min \sum_{(i,j)\in E} w_{ij}x_{ij}&&\label{objective}\\
    s.t.	\sum_{(i,j) \in E}f_{ij} - \sum_{(k,i) \in E}f_{ki} & = &
		    \left\{ 
		    \begin{array}{l l}
		        1  &\quad i=s,\\
		        -1 & \quad i=r,\\
		        0 & \quad \forall i\in V, i\neq s,t,\\
		    \end{array} 
		    \right. \label{flowbalance}\\
            \sum_{(i,j)\in E} t_{i,j}x_{ij}&\leq& t_{max},\label{timeconst}\\
            x_{ij}&\in& \{0,1\}, \quad \forall (i,j)\in E .\label{binary}
\label{equ:ip}            
\end{eqnarray}
\eqref{flowbalance}  ensures that the selected cost-effective devices constitute a path. The time for transmission between devices $i$ and $j$ is obtained considering cellular and the wireless channel as in (\ref{tij}). \eqref{timeconst} makes sure that the selected devices deliver the time-sensitive content within the maximum allowable time $t_{max}$ with high likelihood. This optimization problem is NP-complete since it belongs to a class of combinatorial optimization \cite{crowcroft96}. Therefore, we cannot derive the optimal solution in polynomial time. Next, we introduce the proposed approach to solve the CEDS-MD problem.

\subsection{Solving the Optimization Problem}

We solve the CEDS-MD problem in three steps: (i) relax the IP formulation into a linear program (LP) and solve it, (ii) show that the optimal solution of the LP has at most two fractional paths that will be constructed  and (iii) formulate a new LP by adding new constraints. Then, we keep solving the modified LP until it becomes infeasible. This approach obtains the optimal solution in near polynomial time  by using interior point method in solving the LP. We start by relaxing (\ref{equ:ip}) to obtain the LP:
\begin{eqnarray}
	\min \sum_{(i,j)\in E} w_{ij}x_{ij}&&\label{objectiveL}\\
    s.t.	\sum_{(i,j) \in E}f_{ij} - \sum_{(k,i) \in E}f_{ki} & = &
		    \left\{ 
		    \begin{array}{l l}
		        1  &\quad i=s,\\
		        -1 & \quad i=r,\\
		        0 & \quad \forall i\in V, i\neq s,t,\\
		    \end{array} 
		    \right. \label{flowbalanceL}\\
            \sum_{(i,j)\in E} t_{i,j}x_{ij}&\leq& t_{max},\label{timeconstL}\\
			0 &\le& x_{ij} \le 1 \quad\forall (i,j)\in E. \label{binaryL}
\end{eqnarray}

\subsubsection{Property of LP Solution}
We denote the LP relaxation of \eqref{binaryL} as $\mathcal{P}$. The optimal solution of $\mathcal{P}$ is no longer integral as in the classical shortest path problem\cite{Ahuja88}, due to the addition of constraint \eqref{timeconst}. However, the following theorem holds true. 
\begin{theorem}\label{theorem1}
	There exists either an optimal solution for $\mathcal{P}$ that contains at most two fractional $s,r$ paths or $\mathcal{P}$ is infeasible. 
\end{theorem}

\begin{proof}
	Denote $P_{sr}$ as the collection of all $s,r$ paths. Denote $w_{p_j},t(p_j)$ as the total weight and total delay of a path $ p_j\in P_{sr}$, respectively. $p_j$ is called a long-delay path if $t(p_j)>t_{max}$ and is called a short-delay path otherwise. 
    
    We will show that if $\mathcal{P}$ is feasible and an optimal solution $\textbf{x}^{*}$ contains more than two fractional $s,r$ paths, then either $\textbf{x}^{*}$ can be transformed to an optimal solution with at most two $s,r$ paths or $\textbf{x}^{*}$ is not optimal. Assume $\textbf{x}^{*}$ contains $k>2$ fractional paths and is optimal. It is clear that some short-delay paths must be included, otherwise $\textbf{x}^{*}$ is not even feasible. Therefore, the problem can be categorized into three cases: i) all paths are short-delay paths, ii) at least two short-delay paths and a long-delay path exist and iii) at least two long-delay paths and a short-delay path exist. 
    
    In the first case, if all the short-delay paths selected have the same weight, an equivalent solution can be constructed by assigning flow of 1 to one of the selected paths and flow of 0 to all the others. Such an optimal solution has only one path. If the weight of the selected paths are different, by shifting the flow from heavy-weight paths to light-weight paths can improve the solution and hence, $\textbf{x}^{*}$ is not optimal.
    
    In the second and the third case, the weight of long-delay paths must be smaller than short-delay paths or we can shift the flow to short-delay paths and improve the solution. Denote the collection of all selected paths as $P_{\textbf{x}^{*}}$, we must have
		$\sum_{p_j\in P_{\textbf{x}^{*}}}f_j t(p_j)=t_{max},$
    where $f_j$ is the flow assigned to path $p_j$. If the total time is less than $t_{max}$, it is possible to shift flows from short-delay paths to long-delay paths and improve the solution. 
	In the second case, denote $p_1,p_2$ as two short-delay paths. Also, let $p_a$ as a representation of all other selected paths, where 
    \begin{eqnarray*}
    	f_a&=&\sum_{p_j\in P_{\textbf{x}^{*}},j\neq 1,2}f_j,\\
		t(p_a)&=&\frac{\sum_{p_j\in P_{\textbf{x}^{*}},j\neq 1,2}f_j t(p_j)}{f_a},\\
        w(p_a)&=&\frac{\sum_{p_j\in P_{\textbf{x}^{*}},j\neq 1,2}f_j w(p_j)}{f_a}.
	\end{eqnarray*}
  Clearly,
		$f_a t(p_a) + f_1 t(p_1) + f_2 t(p_2)=t_{max},$\\
        $\text{ }\text{ }\text{ }\text{ }\text{ }\text{ }\text{ }\text{ }\text{ }
        \text{ }f_a w(p_a) + f_1 w(p_1) + f_2 w(p_2) = Y^*,$\\
    where $Y^*$ denotes the objective value of solution $\textbf{x}^{*}$. Also, we have $t(p_a)>t_{max}, w(p_a)<w(p_1), w(p_2)$.
    
    Without loss of generality, let $t(p_1)<t(p_2)$, then $w(p_1)>w(p_2)$ or $p_1$, $p_2$ cannot coexist in the optimal solution. Consider two moves: (1) Remove $p_2$ from the optimal solution. (2) Remove $p_1$ from the optimal solution. For both moves, the solutions are recalculated by assigning flows to the remaining selected paths. Denote the objective value by $Y^1,Y^2$ for move (1) and (2) respectively. We will show that it is impossible to have both $Y^1,Y^2 \leq Y^*$ and $Y^*$ is not an optimal solution. 
    
    After move (1), the following formulas hold.    
    \begin{eqnarray*}
		t_{max}&=&(f_a + \delta_1) t(p_a) + (f_1+f_2-\delta_1) t(p_1),\\
        Y^1&=&(f_a + \delta_1) w(p_a) + (f_1+f_2-\delta_1) w(p_1),\\
        \delta_1&=& f_2\frac{t(p_2)-t(p_1)}{t(p_a)-t(p_1)}.            
	\end{eqnarray*}
    Therefore, 
    \begin{eqnarray*}
		\Delta_1 = Y^1-Y^*=f_2(w(p_1)-w(p_2))+\delta_1(w(p_a)-w(p_1))\\
        =f_2((w(p_1)-w(p_2))+\frac{t(p_2)-t(p_1)}{t(p_a)-t(p_1)}(w(p_a)-w(p_1))).
	\end{eqnarray*}

    After move (2), the following formulas hold. 
    \begin{eqnarray*}
		t_{max}&=&(f_a - \delta_2) t(p_a) + (f_1+f_2+\delta_2) t(p_2),\\
        Y^2&=&(f_a - \delta_2) w(p_a) + (f_1+f_2+\delta_2) w(p_2),\\
        \delta_2&=& f_1\frac{t(p_2)-t(p_1)}{t(p_a)-t(p_2)}.
	\end{eqnarray*}
    Therefore, 
    \begin{eqnarray*}
		\Delta_2 = Y^2-Y^*=f_1(w(p_2)-w(p_1))-\delta_2(w(p_a)-w(p_2))\\
         =f_1((w(p_2)-w(p_1))+\frac{t(p_2)-t(p_1)}{t(p_a)-t(p_2)}(w(p_2)-w(p_a))).
	\end{eqnarray*}
    Assume $\Delta_1, \Delta_2>0$, since $f_1,f_2>0$, we have
    \begin{eqnarray}
		\frac{w(p_1)-w(p_2)}{w(p_1)-w(p_a)}&>&\frac{t(p_2)-t(p_1)}{t(p_a)-t(p_1)},\label{move1}\\
        \frac{w(p_1)-w(p_2)}{w(p_2)-w(p_a)}&<&\frac{t(p_2)-t(p_1)}{t(p_a)-t(p_2)}.\label{move2}
	\end{eqnarray}

    However, inequality \eqref{move1}, \eqref{move2} cannot both hold simultaneously. To see it clearly, let 
    \begin{eqnarray}
		a=w(p_1)-w(p_2),\quad b=w(p_2)-w(p_a),\\
        c=t(p_a)-t(p_2), \quad d=t(p_2)-t(p_1).
	\end{eqnarray}
	Then, inequality \eqref{move1} reduces to $\frac{a}{a+b}>\frac{d}{c+d}$, while inequality \eqref{move2} reduces to $\frac{a}{b}<\frac{d}{c}$. The first one implies $ac>bd$ while the second one implies $ac<bd$.
    
    Therefore, in the second case in which there exists two or more short-delay paths in the solution, we can always perform move (1) or (2) to reduce number of short-delay paths without increasing objective value. The same claim holds true for the third case with a similar reasoning. 
    
    In conclusion, we can always create an optimal solution for $\mathcal{P}$ while selecting at most two $s,r$ paths.
\end{proof}

\subsection{Exact Solution by Cutting Plane}
Based on Theorem \ref{theorem1}, an optimal solution with at most two fractional paths can always be generated by solving $\mathcal{P}$. The case of only one path is trivial since it is already the optimal integral solution and no further work is required. Therefore, we are only interested in solutions with two fractional paths. Clearly, the two paths must be one short-delay path, denoted as $p_s$ and one long-delay path, denoted as $p_l$. Since any feasible integral solution must be a short-delay path, we are particularly interested in $p_s$. Denote $X_{p_s}=\sum_{(i,j)\in p_s}\bar{x}_{ij}$, where $\bar{x}_{ij}$ is the value of $x_{ij}$ in the current solution. If we cut the path $p_s$ out of the feasible region of $\mathcal{P}$, the solution must explore other paths by adding the following constraint
	\begin{eqnarray}
		\sum_{(i,j)\in p_s}x_{ij}<X_{p_s} \label{cuttingplane}
	\end{eqnarray}
    
    By resolving $\mathcal{P}$ iteratively while updating the constraint \eqref{cuttingplane}, the feasible region of $\mathcal{P}$ is gradually decreased. We continue the iteration until it is infeasible.
   
The optimal solution will then be the short-delay path with minimum weight. The final algorithm, which we call relay path finder ({\bf RPF}), is presented in Alg. \ref{algRPF}.

\begin{algorithm}[t]
\footnotesize
  \caption{RPF: An optimal algorithm for finding least cost relay path}
  \KwData{Network $G_r = (V_r,E_r,W_v,W_e)$, source $s$, target $r$ and $t_{max}$}
  \KwResult{A path comprising a set $S$ of edges forming the relay}
  Initialize $Q \leftarrow \emptyset$  \;\\
  Solve the LP in \eqref{objective}-\eqref{binaryL} \;\\  
  $P \leftarrow \text{solution of LP}$\; \\  
  \While{$P \text{ is feasible}$}
  {
    $ F \leftarrow \text{ \{construct feasible path(s) in P\} }$\;\\
	$Q \leftarrow Q \cup\text{ \{short delay path in } F\}$ \;\\ 
	\If{$F \text{ contains only one path from s to r}$}
	{
	    Return the path in $Q$ with smallest weight
	}	  	
	\ElseIf{$F \text{ contains two paths from s to r}$}
	{
		Let $p_s$ and $p_l$ be the paths\;\\
		Add constraint according to \eqref{cuttingplane} to the LP 
	}
	Solve the LP with the additional constraint\;\\
	$P \leftarrow \text{solution of the updated LP}$
  }
  \If{$P \text{ is infeasible \&\& } Q = \emptyset $}
  {
      No feasible path exists\;\\	  
  	  Initiate direct cellular communication between BS and $r$
  }
  \ElseIf{$Q <> \emptyset$ \&\& $\{\exists P'\in Q | C(P') < C(B2D)\}$}
  {
 	  Return the path in $Q$ with smallest weight and cost $<$ B2D
  }  
  \Else
  {
  	Initiate direct cellular communication between BS and $r$
  }

  \label{algRPF}
\end{algorithm}

\section{Performance Evaluation}\label{label:experiment}
For our simulations, the mobility trace for nodes is generated by self-similar least action walk model (SLAW) which is shown to be very realistic in capturing user mobility \cite{slaw}. In particular, SLAW
generated traces are shown to be effective in representing social contexts present among people sharing common interests or those in a single community such as university campus, companies and theme parks. In human mobility, people strive to reduce the distance of travel by visiting all the nearby destinations before visiting farther destinations unless some high priority events such as appointments force them to make a long distance trip even in the presence of unvisited nearby destinations. SLAW leverages this self-similarity of fractal waypoints, which can be viewed as destinations, to realistically predict the human mobility. In this paper, we have used the similar parameter settings for capturing this regularity in human mobility patterns which are also suggested in the original paper \cite{slaw}. The wireless propagation channel is modeled for urban macrocell scenarios with shadowing component set to having standard deviation of $12$ dB and path loss exponent $\alpha$ set to $3$. The cell area is set up as a $1$~km $\times$ $1$~km square with the BS at its center. The noise spectral density is $-174$~dBm/Hz. The transmit power for each device is $100$~mW whereas the power of the BS is set to $10$~W. The total bandwidth of the RBs are set to $5$~MHz in accordance with LTE RBs~\cite{3gpp} and the maximum D2D distance is set to $d_{max} = 15$~m. The main wireless network parameters are listed in Table \ref{table:wirelessparam}. We have set $\rho$, $\zeta$ and $\delta$ to $0.8$, $0.7$ and $4$ respectively in constructing $G_p$ for durable community detection. We describe how to choose these values later in this section.

\begin{center}
\begin{table}[!htb]
	\centering	
	\caption{Main Wireless Network Parameters} 
	\begin{tabular}{|l|l|}
		\hline Notation & Description \\
         \hline Cell dimension & $1000$ x $1000$ $m^2$ \\ 
         BS location & Center of the area \\
         Shadowing std. dev. & $12$ dB \\ 
         Path Loss Exponent & $3$ \\ 
         Noise spectral density & $-174$ dBm/Hz \\ 
		 BS transmit power& $10$ W \\ 
		 D2D transmit power & $100$ mW \\ 
         Maximum D2D distance & $15$ m \\ 
		 RB size & $12$ sub-carriers, $0.5$ ms \\ 		        
		\hline
	\end{tabular}
	\label{table:wirelessparam}
\end{table}
\end{center}

We have compared the performance of our solution, RPF, with Groups-NET (GNET in short) which is a mobility-aware social-based approach that analyses the impact of device mobility on the cellular network performance and multi-hop D2D in particular \cite {gnet}. GNET identifies social groups based on previous social meetings. It then computes the likelihood of each group meeting in future by computing the group-to-group paths by considering the meeting regularity and shared group members. Finally, it identifies the most probable path from the source to the destination by leveraging the group-to-group path probability. It has been shown that GNET outperforms other state-of-the-art methods in terms of improving the cellular network efficiency \cite{gnet}. We also compare our results with two other social-oblivious methods: i) minimizing cost (MC) scheme that chooses devices that minimize the cost of the BS in content transmission, ii) closest to destination (CD) scheme that selects the device that is physically closest to destination at each hop. These greedy methods have been used for relay selection in multi-hop D2D as an efficient way to offload cellular traffic and to enable content transfer through D2D when direct connection can not be established between the source and the destination.

We generated location of total 400 users in the designated area using SLAW model for 72 hours and used first ${\Delta}t=48$ hours for detecting the social encounter based communities. The rest 24 hours were used for simulating the D2D content transfer. We randomly chose 20 cellular users uniformly distributed over the area and 20 pairs of D2D devices as source and target (having distance larger than $d_{max}$) and averaged the results over a large number of independent simulation runs.

\begin{figure*}[!ht]
\centering
  \subfigure[Impact of $t_{max}$ and different content size for RPF] {  		  
		\includegraphics[scale=0.35]{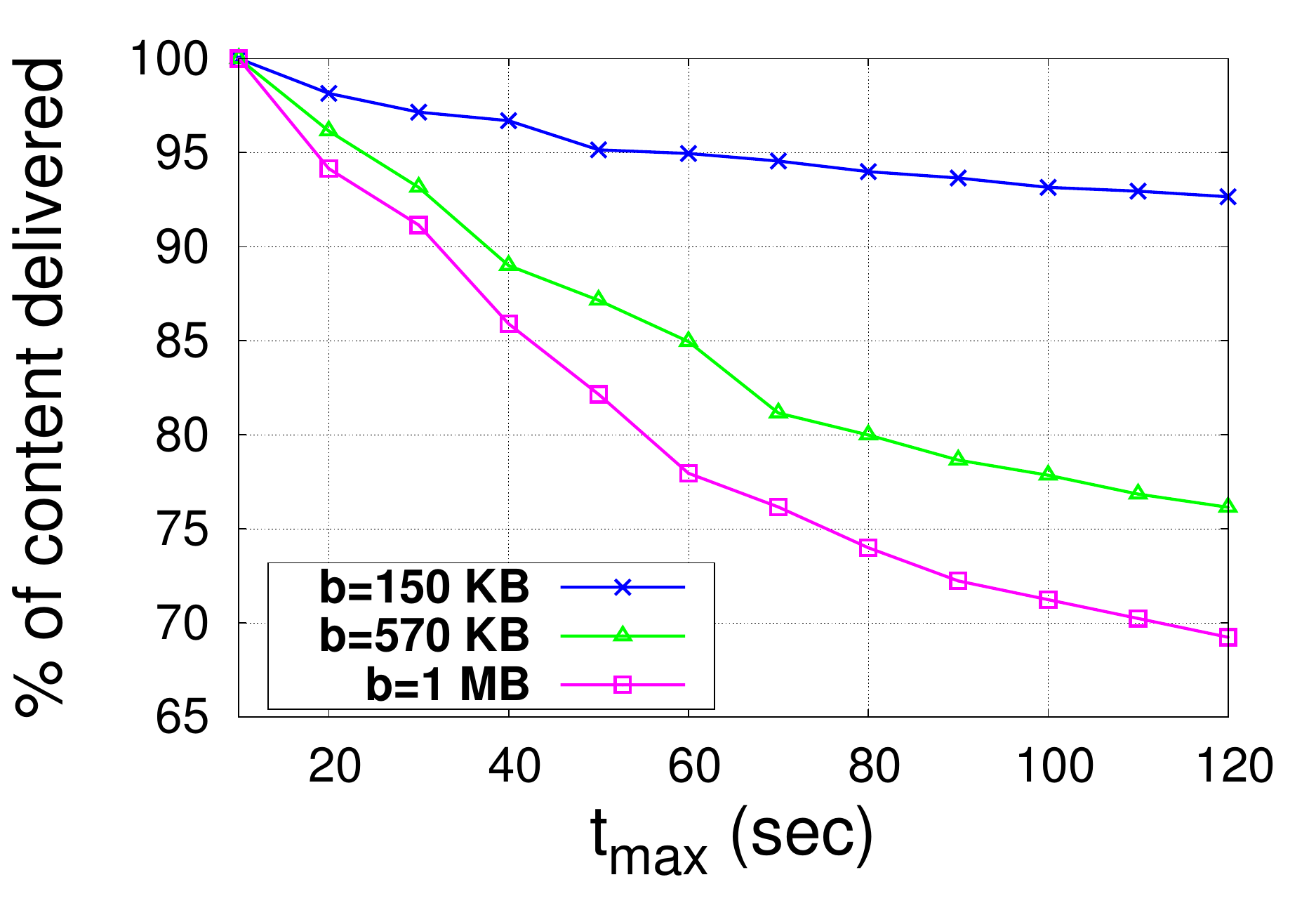}
		\label{fig:acctmax}
	}
	\subfigure[Impact of content size $b$] {
		\includegraphics[scale=0.35]{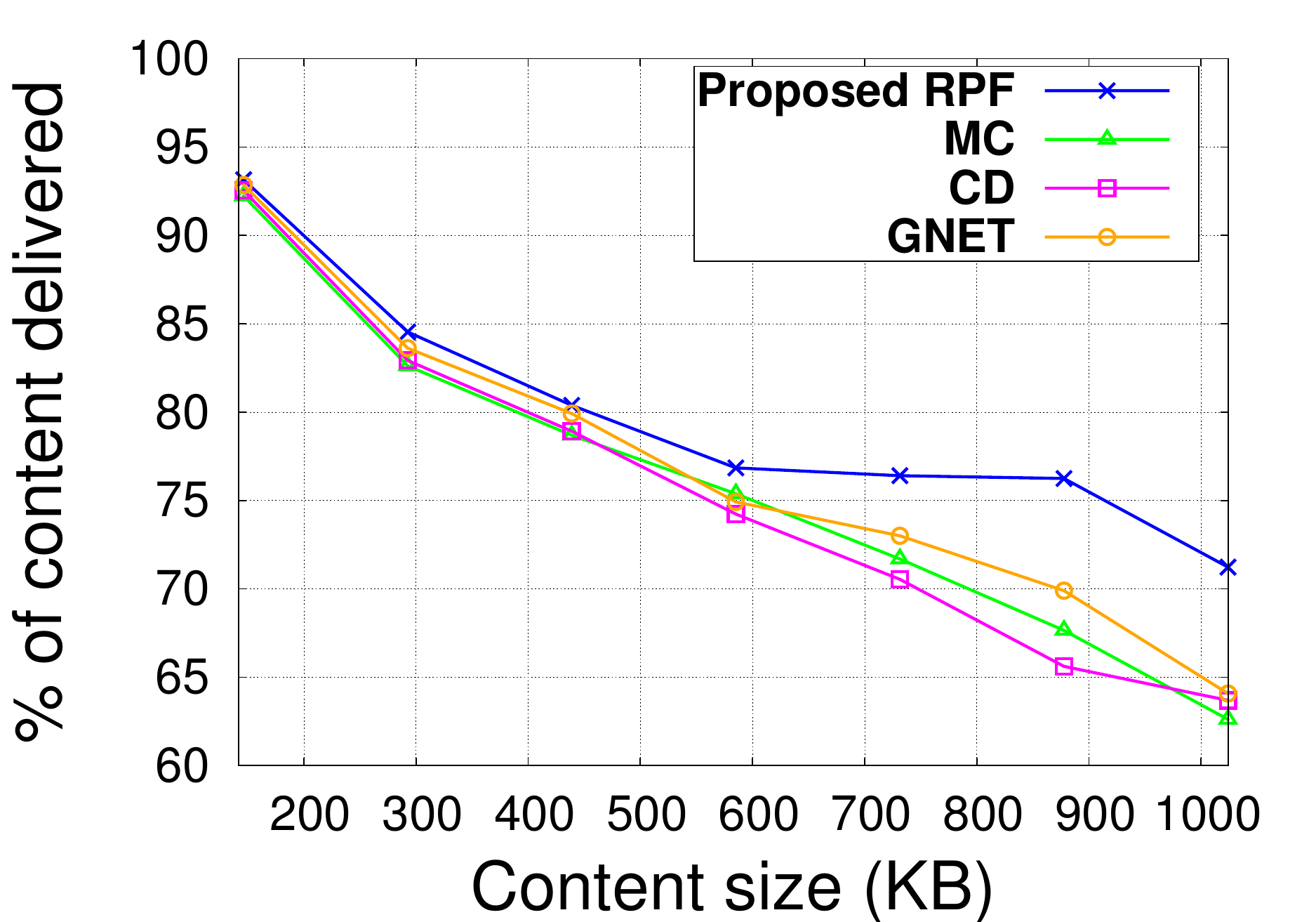}
		\label{fig:accb}
	}
\subfigure[Impact of total users] {
		\includegraphics[scale=0.35]{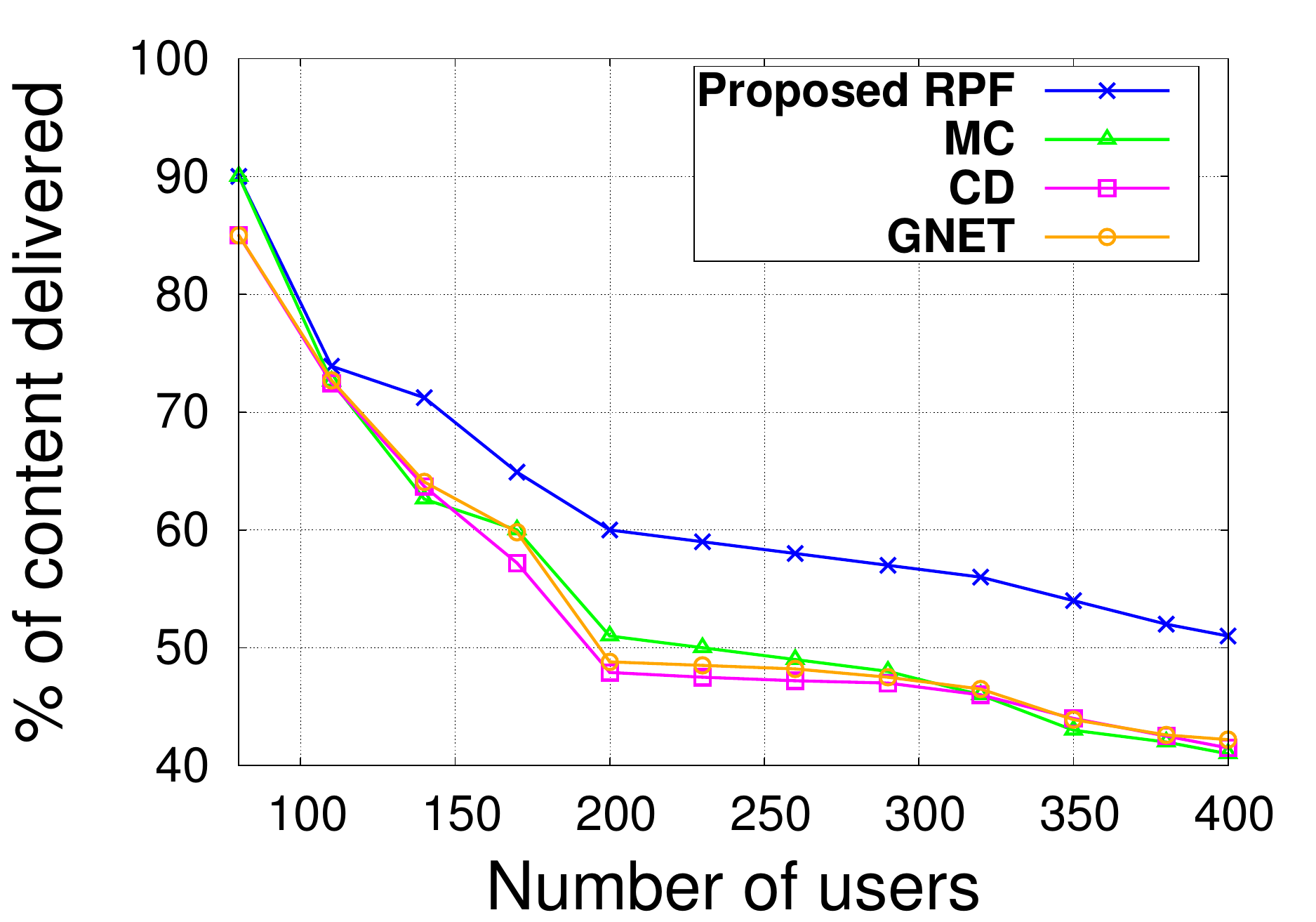}
		\label{fig:accuser}
	}	
  \caption{Content transmission success rate for different cases}
	\label{fig:accuracy}
\vspace{.5cm}
\end{figure*}

Fig. \ref{fig:accuracy} compares the content delivery rate for the proposed algorithm and the baseline approaches when different parameters are varied. In Fig. \ref{fig:acctmax}, we show the content delivery rate achieved by our proposed algorithm RPF for $140$ users and three different content sizes $150$ KB, $570$ KB, and $1$ MB as the content sharing time $t_{max}$ is varied from $10$~s to $120$~s. For a particular content size $b$, with increasing $t_{max}$, the RPF tends to choose devices on the multi-hop path with delivery time close to $t_{max}$ so as to minimize the cost. This results in more hops on a multi-hop D2D path, thus making it more susceptible to device mobility. Consequently, the content delivery success rate keeps decreasing with larger values of $t_{max}$. However, RPF chooses the same multi-hop path  after certain value of $t_{max}$ as the path cost can no longer be minimized within $t_{max}$. From the Fig. \ref{fig:acctmax}, we can see that the delivery success rate remains at around $70\%$ after $t_{max}$ reaches $120$~s for content size $b=1$ MB. For a particular $t_{max}$ value, the delivery success rate decreases with larger content size. Larger content requires more time to be transmitted  which makes them more prone to device mobility. Therefore, a larger content size results in reduced content delivery rate for a fixed $t_{max}$ which is also evident from the Fig. \ref{fig:acctmax}.

In Fig. \ref{fig:accb}, we show how varying the content size impacts the content delivery rate when $t_{max}$ is set to $100$~s for $140$~users. Clearly, as content size increases, the delivery rate decreases. However, the rate of degradation for RPF is much smaller than other methods. This is due to the fact that larger contents require more time for transmission which, in turn, makes the longer D2D session more susceptible to device mobility. In such cases, the mobility of the devices can lead to premature tear down of the multi-hop D2D session. As a result, methods that do not account for social communities in choosing reliable devices on multi-hop will experience a poor delivery rate. Fig. \ref{fig:accb} shows that the content delivery rate is up to $14\%$ higher for the proposed RPF algorithm when compared to the social-unaware scenario for $b=1$~MB.

Fig. \ref{fig:accuser} shows how the content delivery rate varies with the network size. As the number of users increases, one expects a better delivery rate due to more option for multi-hop. However, a large number of users will also increase interference for users that need to transmit on the same RB. In such a scenario with scarce resources, achievable data rate decreases leading to longer transmission time which makes them more susceptible to device mobility. Interestingly, RPF suffers less from the increased user concentration which makes it best device selection method. In Fig. \ref{fig:accuser}, we can see that the proposed RPF is more resilient to mobility than all other approaches. Moreover, the content delivery rate resulting from RPF is up to $24\%$ higher than the social-unaware scenario for a user count of $400$.

Furthermore, from Figs. \ref{fig:acctmax}-\ref{fig:accuser}, we can see that all the baseline methods perform poorly compared to RPF in terms of content delivery success rate. On the one hand, for CD, since it does not consider the signal and noise information, it suffers from poor content delivery. On the other hand, MC always tries to minimize the BS cost which, in some cases, results in choosing devices that require large time to deliver thus making it prone to experience disconnections during mobility. GNET also suffers from poor delivery rate as it prioritizes the most probable community-to-community path. Two adjacent devices belonging to two different communities with large community-to-community path probability will be chosen by GNET, even if they have never met before. As a result, these devices without significant previous meeting records, might move far apart from each other during a transmission session leading to poor delivery of content. On the other hand, RPF's consideration of durable social communities enables it to identify devices that are likely to maintain the required QoS during the whole session by remaining close to one another. 

\begin{figure*}[t]
  \centering  
  \subfigure[Active B2D links vs $t_{max}$] {  		  
		\includegraphics[scale=.35]{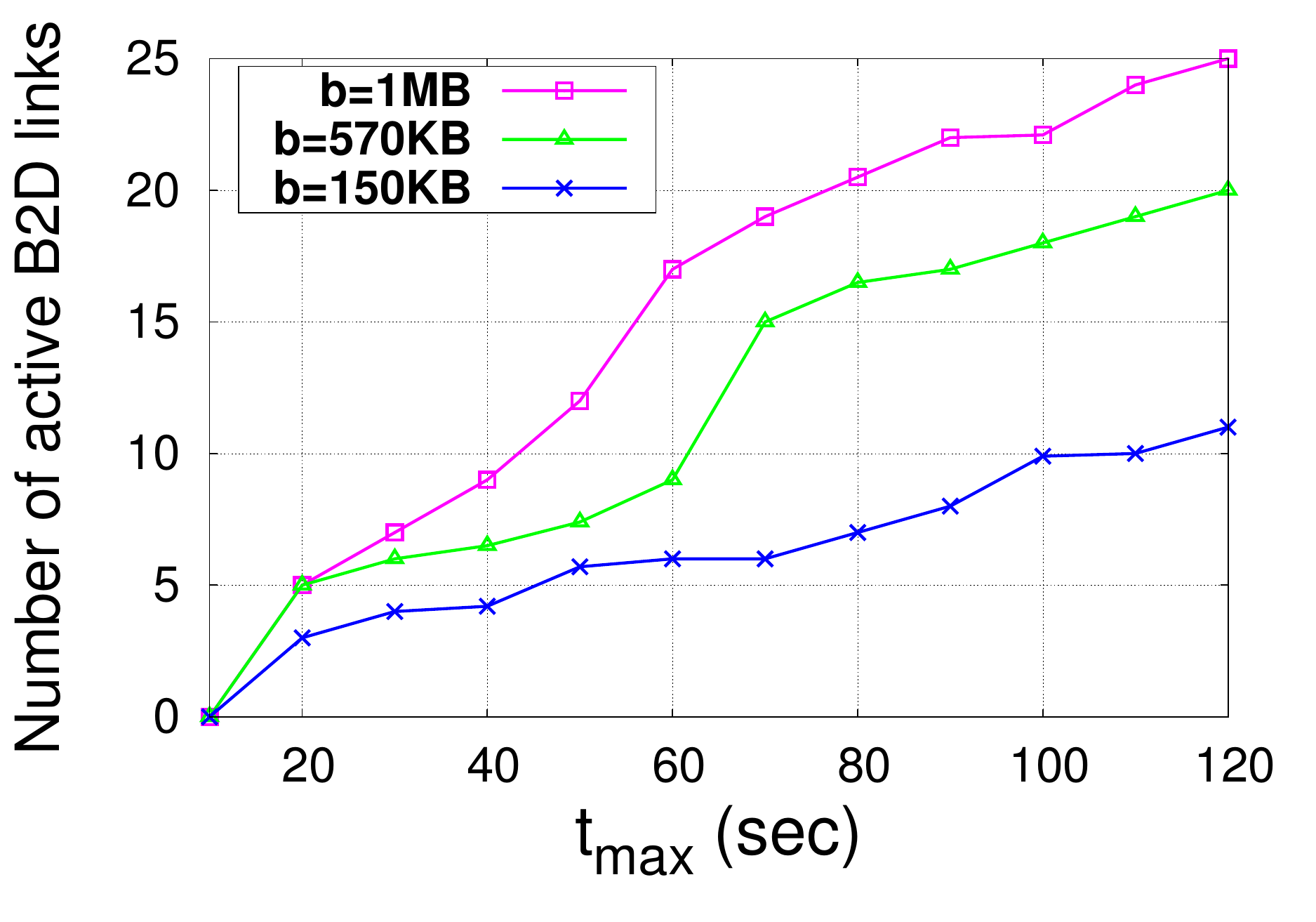}	
	\label{fig:b2dtmax}
	}
	\subfigure[Active B2D links vs SINR] {
		\includegraphics[scale=.35]{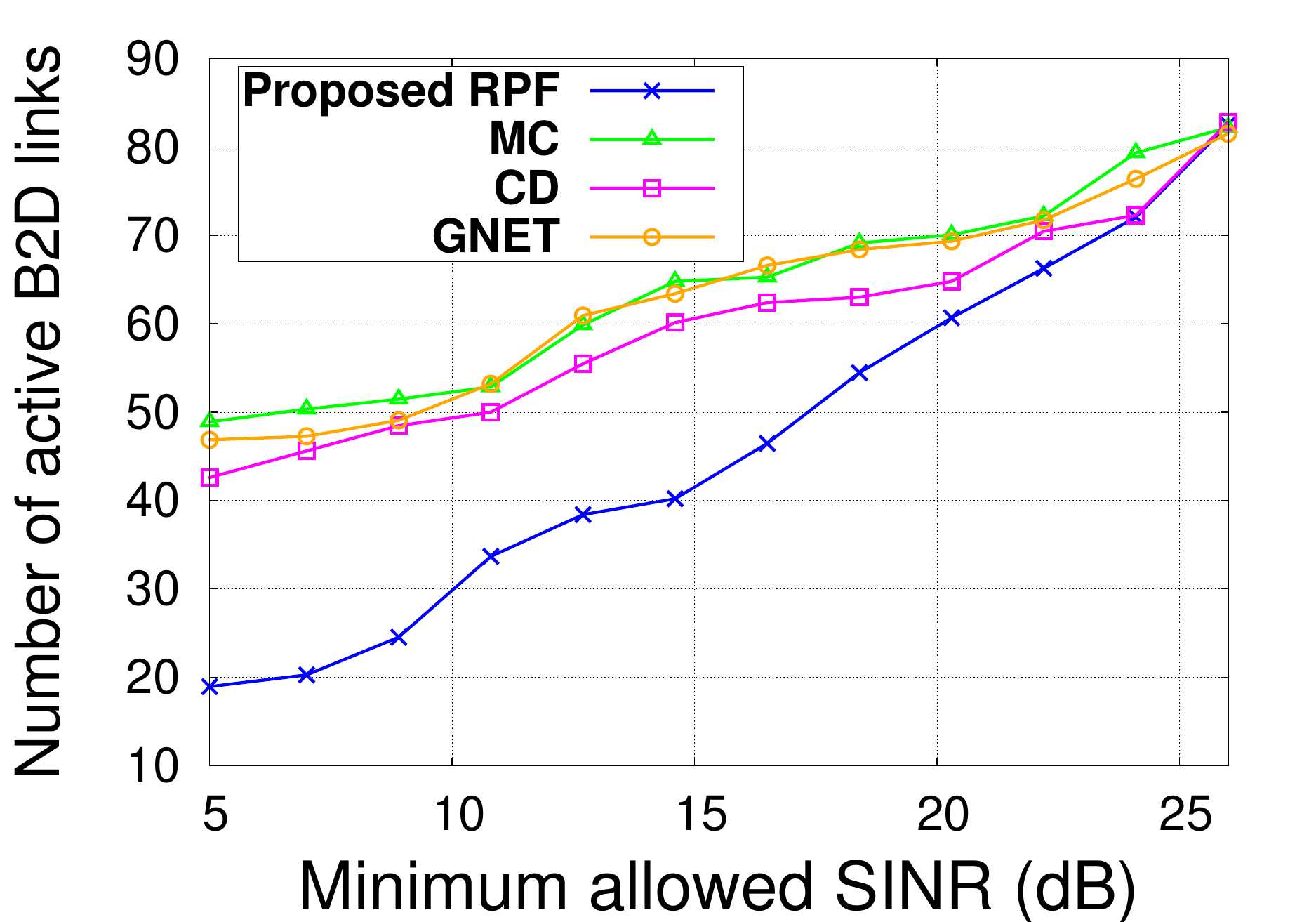}	
  \label{fig:b2ddmax}
	}
  \subfigure[Active B2D links vs content size] {
		\includegraphics[scale=0.35]{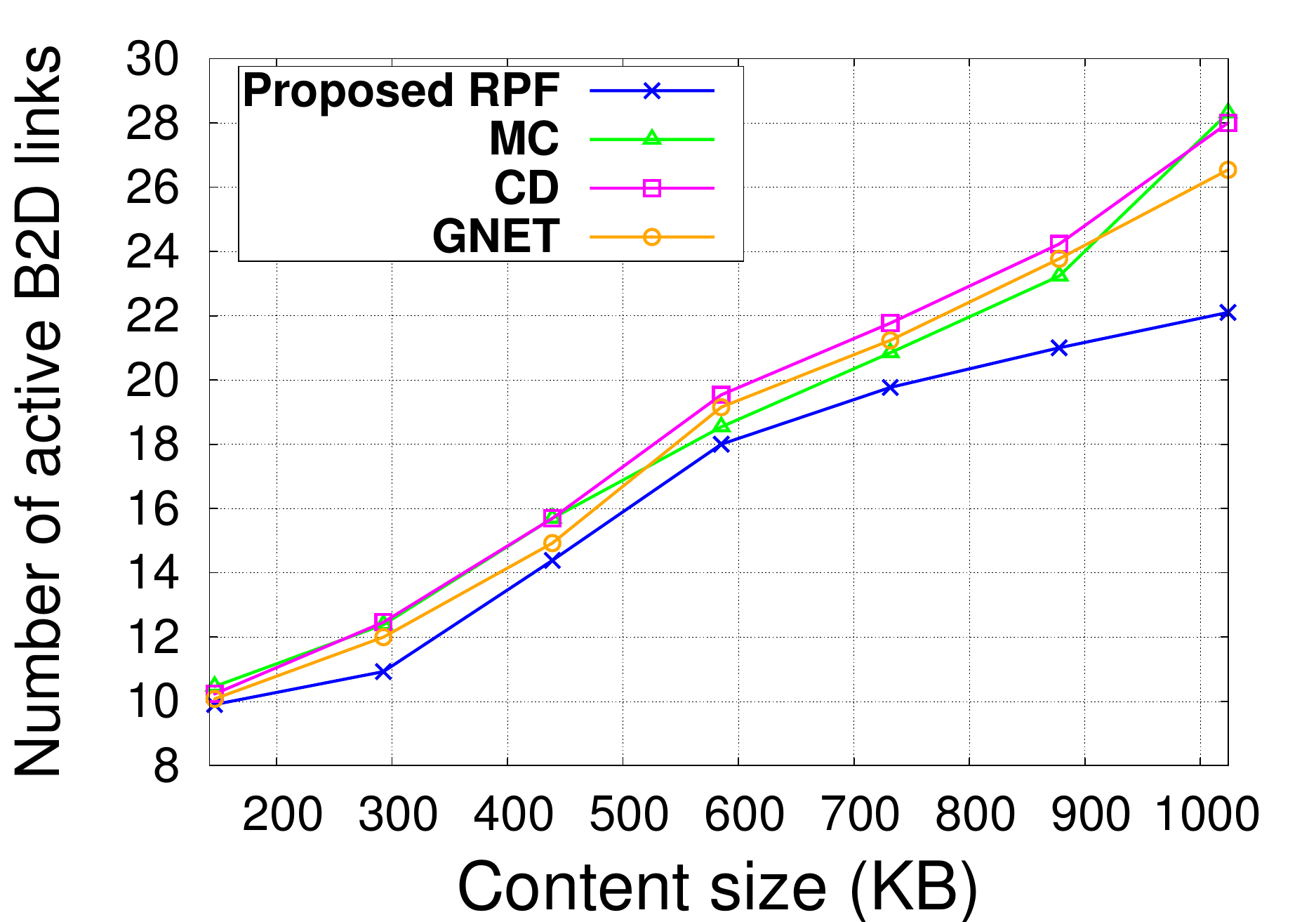}
		\label{fig:b2db}
	}
  \caption{Offload performance analysis for different cases}
	\label{fig:offload}
	\vspace{.5cm}
\end{figure*}

Fig. \ref{fig:offload} evaluates the offload performance of the proposed RPF. For a $100$ seconds duration, we recorded the number of active B2D links in the network which is shown in Figs. \ref{fig:b2dtmax}-\ref{fig:b2db}. The BS initiates a direct cellular connection towards the target when: a) there is no feasible multi-hop path or b) the multi-hop device cost is larger than the direct BS to device (B2D) cost or c) the mobility of devices on a path leads to a premature disconnection of that path.  

Fig. \ref{fig:b2dtmax} shows the impact of increasing $t_{max}$ as contents of three different sizes $150$ KB, $570$~KB, and $1$~MB are transmitted. Since  contents can be transmitted for a longer duration with increasing $t_{max}$, the number of active B2D links increases for each of the content size $b$. This corroborates the intuition that the mobility of devices may disrupt D2D sessions which will require the BS to use costly B2D links. As the content size increases, more time is needed for content transmission which, in turn, makes the multi-hop D2D path more prone to device mobility. Consequently, the premature tear down of an ongoing session due to device mobility leads to the reliance on increasing number of  B2D links where the content is served directly by the BS to the content requester. 

In Fig. \ref{fig:b2ddmax}, we can see the impact of minimum allowed SINR on network traffic offload. In case of smaller SINR, devices can sustain longer D2D sessions since the required QoS for such a communication is low. This results in more  successful content delivery over multi-hop D2D which, in turn, requires less number of B2D links. However, when the allowable SINR is  increased, the tolerance to device mobility is decreased which subsequently results in more active B2D links. From Fig. \ref{fig:b2ddmax}, we can see that the other methods require as high as $158\%$ more B2D links compared to the proposed RPF for a target SINR of $5$~dB.

In Fig. \ref{fig:b2db}, we show the comparative performance of different relay selection methods for varying number of content size from $150$~KB to $1$~MB for a fixed network of size $140$ and $t_{max}=100$~s. As content size increases, all methods will start to increasingly rely on the B2D links. However, RPF requires $28\%$ less B2D
links compared to other methods. The reduction in the number of B2D links demonstrates the improved offload
capabilities of the proposed RPF. Such an offload of traffic from the BS to the D2D tier also
reduces the usage of expensive backhaul traffic.

\begin{figure}
  \centering  
	\includegraphics[scale=.5]{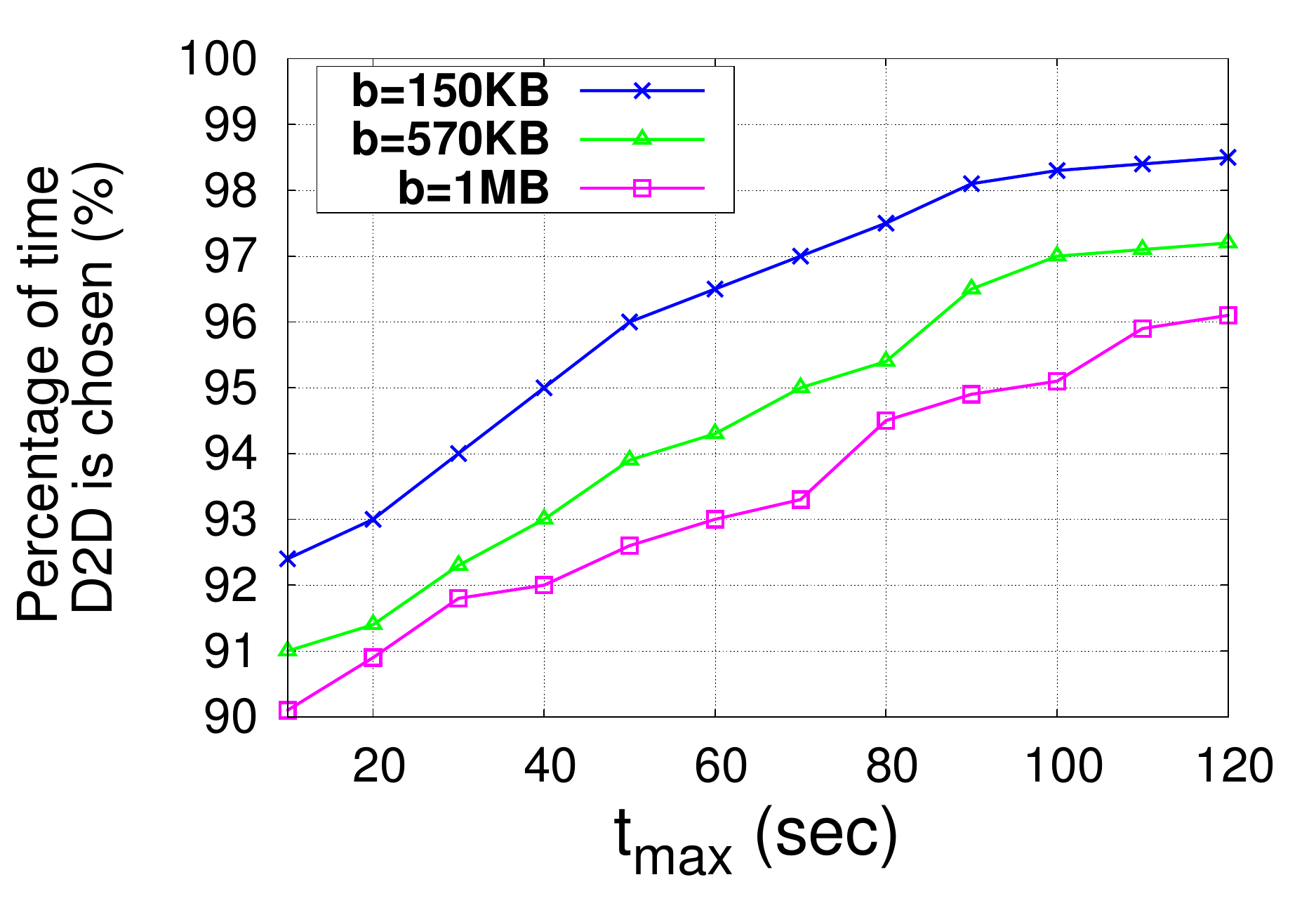}	
\caption{Cost-effectiveness of multi-hop D2D for three different content sizes and a range of $t_{max}$}\label{fig:cost}
\end{figure}

In Fig. \ref{fig:cost}, we show the percentage of time a multi-hop D2D path is chosen instead of an expensive direct B2D link for a user count of $140$. This figure also gives an indication on the quality of the cost functions that we have defined in \eqref{equ:weightNode} and \eqref{equ:BSweight}. Note that, this comparison considers only the cost of direct B2D and D2D relay devices before the transmission starts. Fig. \ref{fig:cost} shows that, over $90\%$ of the time, the RPF chooses multi-hop D2D due to its cost-effectiveness. The small portion of time during which the direct B2D links are used is primarily due to the destination devices which are closer to the BS and can receive the content directly from the BS with low cost. As the allowed $t_{max}$ increases, more D2D links are chosen by RPF compared to B2D links. With increasing $t_{max}$, the RPF tends to choose devices on the multi-hop path with delivery time close to $t_{max}$ as it tries to minimize the cost. This results in a D2D path having a smaller cost, which explains why more D2D links are selected by RPF as the $t_{max}$ increases for a particular content size. Furthermore, as the content size increases, the chances of forming better (less expensive) D2D path starts decreasing. For a given $t_{max}$, RPF has to choose devices of higher cost as the content size $b$ increases in order to find a D2D path that is capable of delivering the content within $t_{max}$. Therefore, a larger content size results in decreased percentage of D2D links being chosen by RPF before the content transmission starts as shown in Fig. \ref{fig:cost}.

Fig. \ref{fig:exectime} shows the execution time needed for our approach. RPF achieves the optimal solution within shortest possible time even in large networks. In almost all realizations, it takes less than a second on the average to compute cost-effective devices on multi-hop path. We performed all the computations on an AMD Opteron(tm) Processor 6168 CPU with 64 GB-memory Linux machine.


In Fig. \ref{fig:heatmap}, we show the impact of different parameters  mentioned in Subsection~\ref{label:lscd} on the performance of the RPF algorithm. Fig. \ref{fig:heatmap} is the heat map representing the impact of $\delta$ (stability threshold) and the weight factor $\rho$ on the content delivery success rate achieved by RPF for a user count of $140$, content size $b=1$~MB, and $t_{max}=100$~s. The success rate is depicted by the RGB colors. As the success rate gets higher, the color becomes lighter in the heat map. From Figure \ref{fig:heatmap}, we can see that the color is lightest, i.e., the contents are successfully delivered, in the top right corner where $\rho=0.8$ and $\delta=4$. The content delivery success rate increases till $\delta=4$ and starts deteriorating as $\delta$ is  increased further. Accordingly, we assign the values of $\rho$ and $\delta$ to $0.8$ and $4$ respectively for this setup of user count, content size and $t_{max}$. We vary the value of strength threshold $\zeta$ from $0.5$ to $0.9$ and choose $\zeta=0.7$ as RPF achieves better content delivery with this setup.

\begin{figure*}[!ht] 
\begin{minipage}[t]{0.5\linewidth}
\centering
\includegraphics[width=1\linewidth]{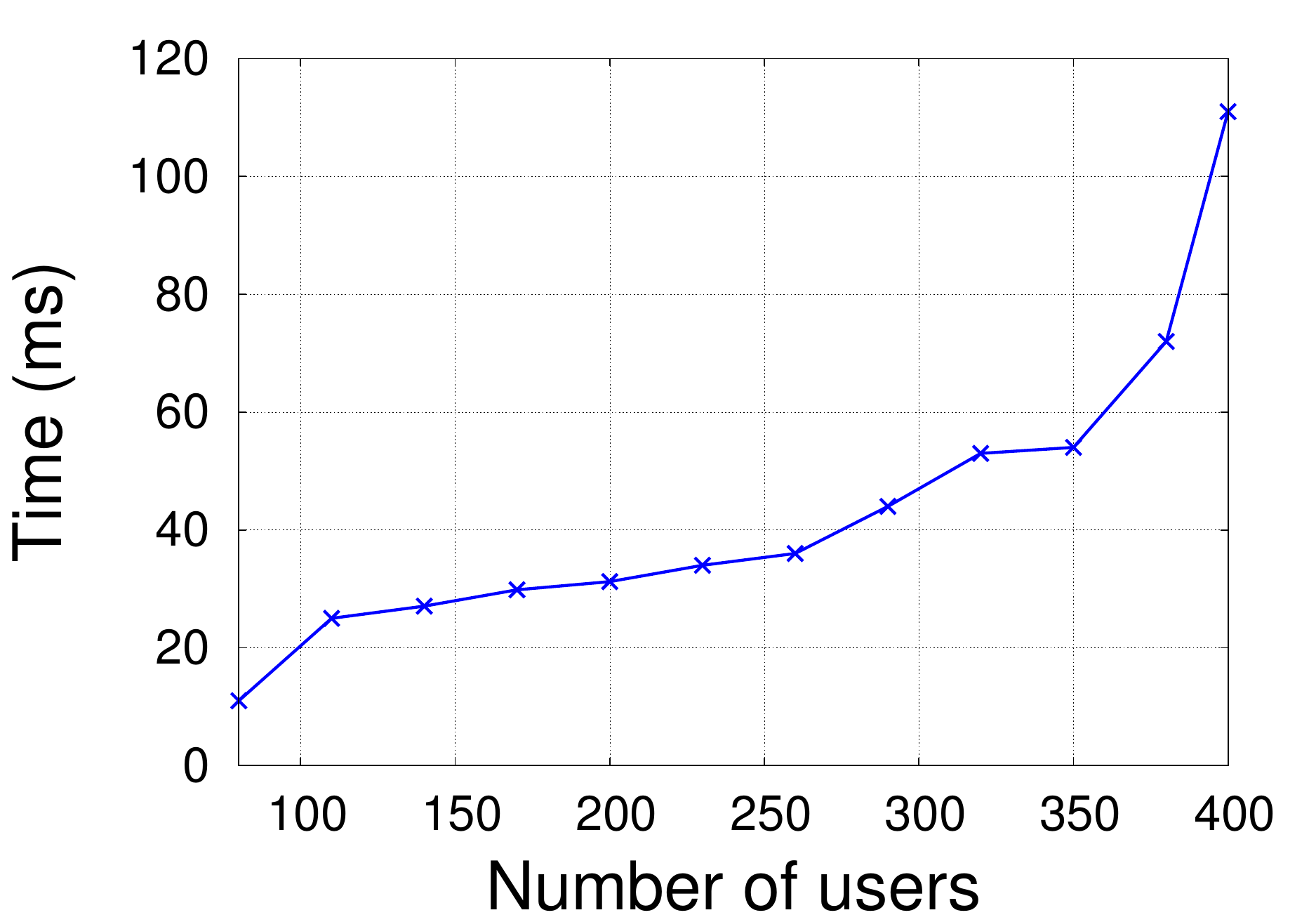}
\caption{Execution time of RPF}
\label{fig:exectime}
\end{minipage}
\begin{minipage}[t]{0.48\linewidth}
\centering
 	\includegraphics[height=0.75\linewidth]{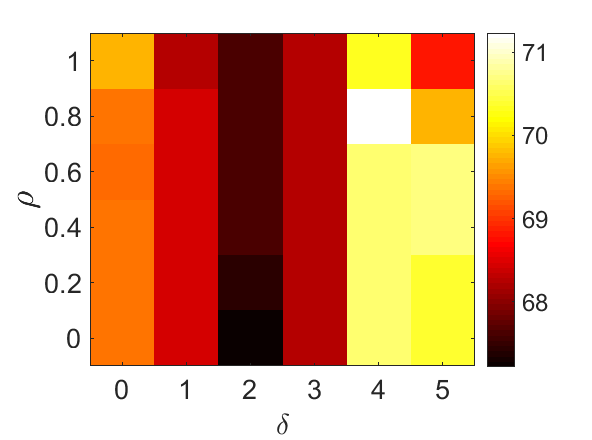}
 	\caption{Impact of different parameters for constructing $G_p$ on the performance of RPF}
 	\label{fig:heatmap}
\end{minipage}
\end{figure*}

\begin{figure}
  \centering  
	\includegraphics[scale=.5]{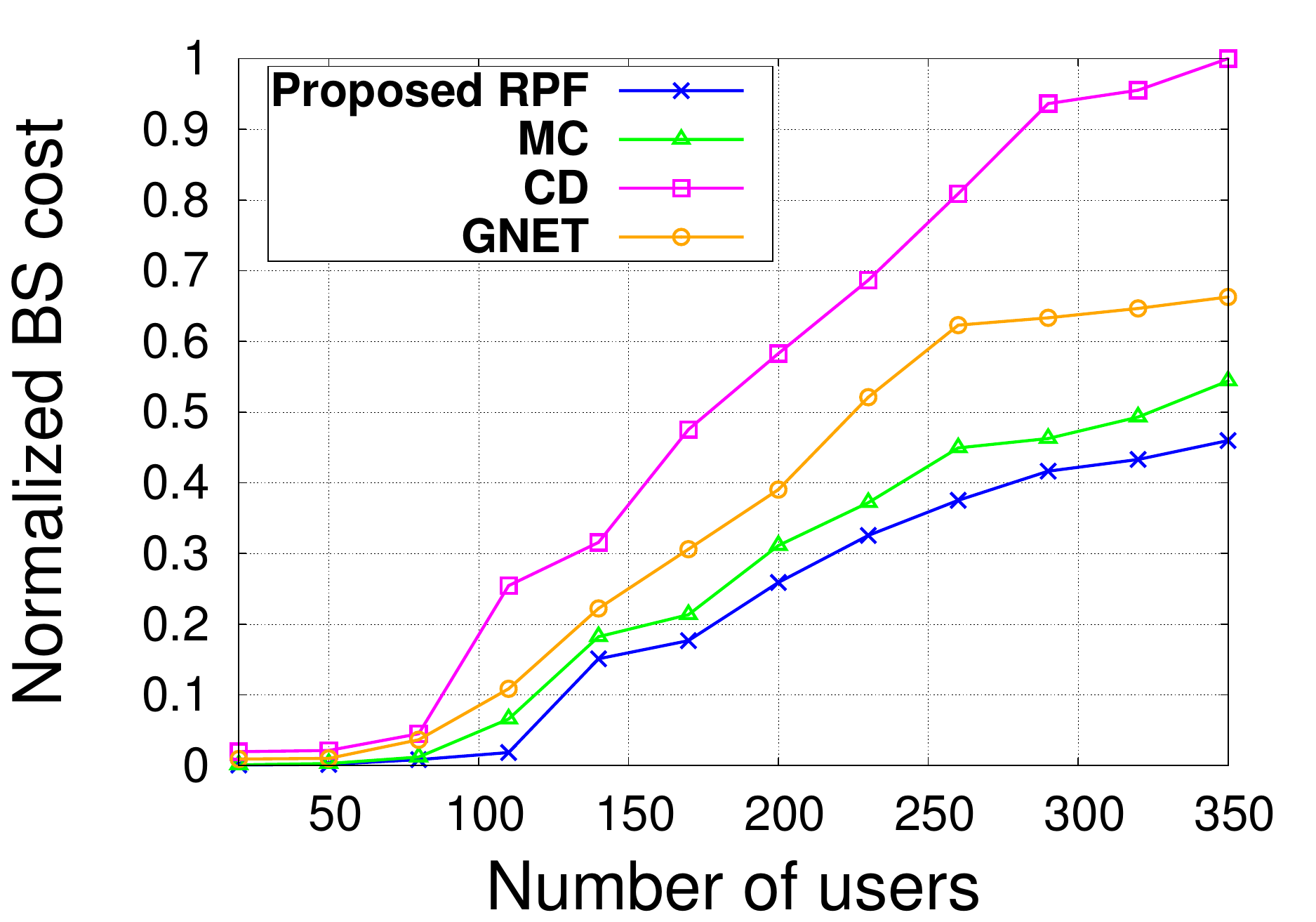}	
\caption{The cost of the BS vs user count}\label{fig:bscost}		  	
\end{figure}

In Fig. \ref{fig:bscost}, we show how the cost of the BS varies with total users in the network. The BS cost is normalized by the highest cost attained for the maximum user count. It is clear that the BS cost for RPF is smaller than that of any other methods. Although, MC aims to choose relay devices that yield a minimum cost, it suffers from poor delivery since it does not take the mobility of devices into account. Therefore, the BS has to invoke expensive B2D links to deliver the content resulting in increased BS cost as evident from the Fig. \ref{fig:bscost}. The other two methods also results in a higher BS cost as they also fail to consider devices' mobility while choosing relay devices. For all of the methods, as the number of users increases so does the interference originating from users that are sharing same resources. In such a scenario, similar to what we have seen in Fig. \ref{fig:accuser}, achievable data rate decreases due to scarce resources. This, in turn, leads to longer transmission time which makes them more susceptible to device mobility and consequently results in higher BS cost. However, as the user count increases, the gap between RPF and other methods also increases validating the superiority of our method in terms of minimizing the BS cost.

\vspace{-16pt}
\section{Conclusion}\label{label:conclusion}
In this paper, we have studied the impact of device mobility on the performance of multi-hop D2D underlaying cellular network. We have introduced a novel model that considers durable communities based on the social encounters of devices for predicting the likelihood of devices' proximity. We have formulated the reliable device selection problem as an IP optimization problem and we have proposed an efficient heuristic algorithm to solve it. We have also shown that leveraging social communities can increase the content delivery rate in multi-hop D2D. Simulation results show that our proposed method outperforms classical social-unaware methods significantly in terms of traffic offload. The results also show that the proposed method achieves its objectives with manageable computational complexity which makes it applicable to larger networks.  

\bibliographystyle{IEEEtran}
\bibliography{ref1}

\begin{thebibliography}{10}
\providecommand{\url}[1]{#1}
\csname url@samestyle\endcsname
\providecommand{\newblock}{\relax}
\providecommand{\bibinfo}[2]{#2}
\providecommand{\BIBentrySTDinterwordspacing}{\spaceskip=0pt\relax}
\providecommand{\BIBentryALTinterwordstretchfactor}{4}
\providecommand{\BIBentryALTinterwordspacing}{\spaceskip=\fontdimen2\font plus
\BIBentryALTinterwordstretchfactor\fontdimen3\font minus
  \fontdimen4\font\relax}
\providecommand{\BIBforeignlanguage}[2]{{%
\expandafter\ifx\csname l@#1\endcsname\relax
\typeout{** WARNING: IEEEtran.bst: No hyphenation pattern has been}%
\typeout{** loaded for the language `#1'. Using the pattern for}%
\typeout{** the default language instead.}%
\else
\language=\csname l@#1\endcsname
\fi
#2}}
\providecommand{\BIBdecl}{\relax}
\BIBdecl

\bibitem{han12}
B.~Han, P.~Hui, V.~S.~A. Kumar, M.~V. Marathe, J.~Shao, and A.~Srinivasan,
  ``Mobile data offloading through opportunistic communications and social
  participation,'' \emph{IEEE Trans. Mobile Computing}, vol.~11, no.~5, pp.
  821--834, May 2012.

\bibitem{fodor12}
G.~Fodor, E.~Dahlman, G.~Mildh, S.~Parkvall, N.~Reider, G.~Miklos, and
  Z.~Turanyi, ``Design aspects of network assisted device-to-device
  communications,'' \emph{IEEE Comm. Mag}, vol. 50(3), pp. 170--177, Mar 2012.

\bibitem{pei2013resource}
Y.~Pei and Y.~Liang, ``Resource allocation for device-to-device communications
  overlaying two-way cellular networks,'' \emph{IEEE Transactions on Wireless
  Communications}, vol.~12, no.~7, pp. 3611--3621, 2013.

\bibitem{madan2010cell}
R.~Madan, J.~Borran, A.~Sampath, N.~Bhushan, A.~Khandekar, and T.~Ji, ``Cell
  association and interference coordination in heterogeneous {LTE-A} cellular
  networks,'' \emph{IEEE Journal on Selected Areas in Communications}, vol.~28,
  no.~9, pp. 1479--1489, 2010.

\bibitem{lin2000multihop}
Y.~Lin and Y.~Hsu, ``Multihop cellular: A new architecture for wireless
  communications,'' \emph{In Proc. of IEEE International Conference on Computer
  Communications}, vol.~3, pp. 1273--1282, 2000.

\bibitem{kaufman2013spectrum}
B.~Kaufman, J.~Lilleberg, and B.~Aazhang, ``Spectrum sharing scheme between
  cellular users and ad-hoc device-to-device users,'' \emph{IEEE Transactions
  on Wireless Communications}, vol.~12, no.~3, pp. 1038--1049, 2013.

\bibitem{asadi2014survey}
A.~Asadi, Q.~Wang, and V.~Mancuso, ``A survey on device-to-device communication
  in cellular networks,'' \emph{IEEE Communications Surveys \& Tutorials},
  vol.~16, no.~4, pp. 1801--1819, 2014.

\bibitem{ma2012distributed}
X.~Ma, R.~Yin, G.~Yu, and Z.~Zhang, ``A distributed relay selection method for
  relay assisted device-to-device communication system,'' \emph{In Proc. of
  23rd International Symposium on Personal Indoor and Mobile Radio
  Communications}, pp. 1020--1024, 2012.

\bibitem{lee2012performance}
D.~Lee, S.~Kim, J.~Lee, and J.~Heo, ``Performance of multihop
  decode-and-forward relaying assisted device-to-device communication
  underlaying cellular networks,'' \emph{In Proc. of International Symposium on
  Information Theory and its Applications}, pp. 455--459, 2012.

\bibitem{vanganuru2012systemmilcom}
K.~Vanganuru, S.~Ferrante, and G.~Sternberg, ``System capacity and coverage of
  a cellular network with {D2D} mobile relays,'' \emph{In Proc. of Military
  Communications Conference}, 2012.

\bibitem{saadincentiveJSAC}
Y.~Zhang, L.~Song, W.~Saad, Z.~Dawy, and Z.~Han, ``Contract-based incentive
  mechanisms for device-to-device communications in cellular networks,''
  \emph{IEEE Journal on Selected Areas in Communications}, vol.~33, no.~10, pp.
  2144--2155, 2015.

\bibitem{wang2012wicom}
L.~Wang, T.~Peng, Y.~Yang, and W.~Wang, ``Interference constrained relay
  selection of {D2D} communication for relay purpose underlaying cellular
  networks,'' \emph{In Proc. of 8th International Conference on Wireless
  Communications, Networking and Mobile Computing}, 2012.

\bibitem{hui2011bubble}
P.~Hui, J.~Crowcroft, and E.~Yoneki, ``Bubble rap: Social-based forwarding in
  delay-tolerant networks,'' \emph{IEEE Transactions on Mobile Computing},
  vol.~10, no.~11, pp. 1576--1589, 2011.

\bibitem{cho2011friendship}
E.~Cho, S.~Myers, and J.~Leskovec, ``Friendship and mobility: user movement in
  location-based social networks,'' \emph{In Proc. of the 17th ACM SIGKDD
  international conference on Knowledge discovery and data mining}, pp.
  1082--1090, 2011.

\bibitem{hui2008human}
P.~Hui and J.~Crowcroft, ``Human mobility models and opportunistic
  communications system design,'' \emph{Philosophical Transactions of the Royal
  Society of London A: Mathematical, Physical and Engineering Sciences}, vol.
  366, no. 1872, pp. 2005--2016, 2008.

\bibitem{wang2014game}
Q.~Wang, W.~Wang, S.~Jin, H.~Zhu, and N.~T. Zhang, ``Game-theoretic source
  selection and power control for quality-optimized wireless multimedia
  device-to-device communications,'' in \emph{In Proc. of IEEE Global
  Communications Conference (GLOBECOM)}.\hskip 1em plus 0.5em minus 0.4em\relax
  IEEE, 2014, pp. 4568--4573.

\bibitem{tan2011graph}
L.~Tan, Z.~Feng, W.~Li, Z.~Jing, and T.~A. Gulliver, ``Graph coloring based
  spectrum allocation for femtocell downlink interference mitigation,''
  \emph{In Proc. of Wireless Communications and Networking Conference (WCNC)},
  pp. 1248--1252, 2011.

\bibitem{proebster2012context}
M.~Proebster, M.~Kaschub, T.~Werthmann, and S.~Valentin, ``Context-aware
  resource allocation for cellular wireless networks,'' \emph{EURASIP Journal
  on Wireless Communications and Networking}, vol. 2012, no.~1, pp. 1--19,
  2012.

\bibitem{brandes2008modularity}
U.~Brandes, D.~Delling, M.~Gaertler, R.~G{\"o}rke, M.~Hoefer, Z.~Nikoloski, and
  D.~Wagner, ``On modularity clustering,'' \emph{IEEE Transactions on Knowledge
  and Data Engineering}, vol.~20, no.~2, pp. 172--188, 2008.

\bibitem{cplex}
\BIBentryALTinterwordspacing
``{IBM ILOG CPLEX Optimization Studio},'' 2014. [Online]. Available:
  \url{http://www-03.ibm.com/software/products/en/ibmilogcpleoptistud}
\BIBentrySTDinterwordspacing

\bibitem{crowcroft96}
Z.~Wang and J.~Crowcroft, ``Quality-of-service routing for supporting
  multimedia applications,'' \emph{IEEE Journal on Selected Areas in
  Communications}, vol.~14, no.~7, pp. 1228--1234, 1996.

\bibitem{Ahuja88}
R.~K. Ahuja, T.~L. Magnanti, and J.~B. Orlin, ``Network flows,'' \emph{DTIC
  Document}, 1988.

\bibitem{slaw}
K.~Lee, S.~Hong, S.~J. Kim, I.~Rhee, and S.~Chong, ``Slaw: A new mobility model
  for human walks,'' \emph{In Proc. of IEEE International Conference on
  Computer Communications}, pp. 855--863, 2009.

\bibitem{3gpp}
\BIBentryALTinterwordspacing
3GPP, \emph{LTE-Advanced (3GPP Release 10 and beyond)}, no. 36.300, Oct 2011.
  [Online]. Available: \url{http://www.3gpp.org}
\BIBentrySTDinterwordspacing

\bibitem{gnet}
I.~O. Nunes, P.~O.~V. de~Melo, and A.~A. Loureiro, ``Leveraging {D2D} multi-hop
  communication through social group meetings awareness,'' \emph{IEEE Wireless
  Communications Magazine}, pp. 1--9, Aug 2016.

\end{thebibliography}

\end{document}